\documentclass[a4]{article}
\usepackage[utf8]{inputenc}
\usepackage[margin=1in]{geometry}
\usepackage{abstract}

\usepackage{blindtext}

\usepackage{amssymb, amsmath, amsthm, amsfonts} 
\usepackage{mathtools}
\usepackage[nice]{nicefrac}

\usepackage{bm}

\usepackage{microtype}

\usepackage{tikz,pgfplots}
\usepackage{tikzscale}
\usetikzlibrary{patterns, shapes}

\usepackage[numbers,compress]{natbib}
\usepackage{booktabs}

\usepackage{xcolor}
\definecolor{darkblue}{rgb}{0,0,0.9}

\usepackage{todonotes}

\newtheorem{theorem}{Theorem}
\newtheorem{lemma}{Lemma}
\newtheorem{prop}{Property}

\newcommand{\eq}[1]{(\ref{eq:#1})}
\newcommand{\fig}[1]{Figure~\ref{fig:#1}}
\newcommand{\secRef}[1]{Section~\ref{sec:#1}}
\newcommand{\lem}[1]{Lemma~\ref{lem:#1}}

\newcommand{\nn}{\nonumber}
\newcommand{\sm}{\setminus}
\newcommand{\ssm}{\smallsetminus}

\newcommand{\mc}[1]{\mathcal{#1}}

\title{Generalised Measures of Multivariate Information Content}
\author{Conor Finn and Joseph T.\ Lizier}
\date{\today}

\begin{document}

\maketitle

\begin{abstract}
  The entropy of a pair of random variables is commonly depicted using a Venn diagram.  This
  representation is potentially misleading, however, since the multivariate mutual information can
  be negative.  This paper presents new measures of multivariate information content that can be
  accurately depicted using Venn diagrams for any number of random variables.  These measures
  complement the existing measures of multivariate mutual information and are constructed by
  considering the algebraic structure of information sharing.  It is shown that the distinct ways in
  which a set of marginal observers can share their information with a non-observing third party
  corresponds to the elements of a free distributive lattice.  The redundancy lattice from partial
  information decomposition is then subsequently and independently derived by combining the
  algebraic structures of joint and shared information content.
\end{abstract}

\section{Introduction}
\label{sec:intro}

\begin{figure}[t]
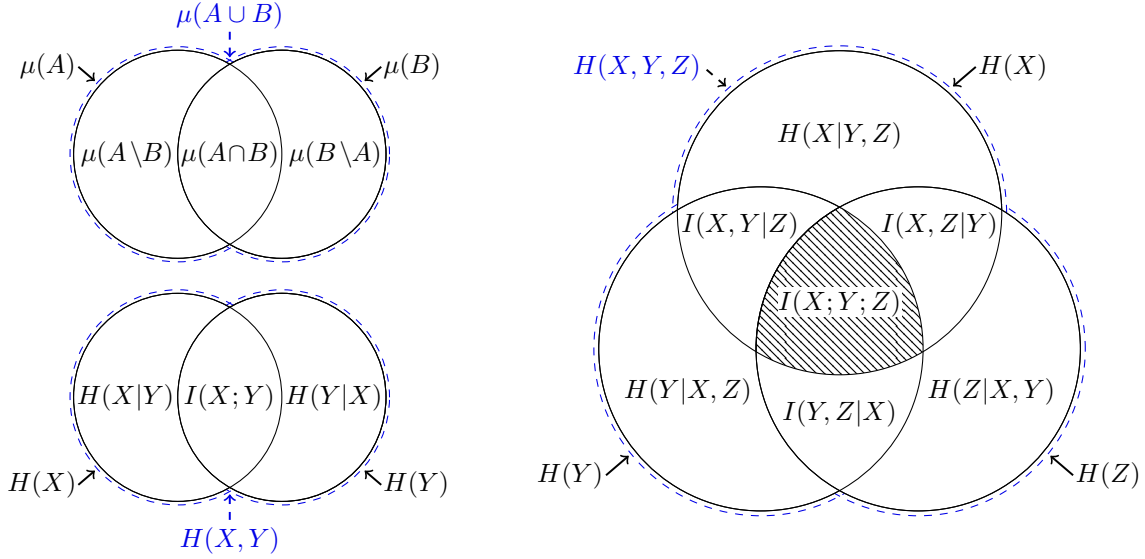

  \centering
  \begin{minipage}{0.37\columnwidth}
    \includegraphics[width=\columnwidth]{figs/two_sets.tikz}
    \vskip -20pt
    \includegraphics[width=\columnwidth]{figs/two_entropies.tikz}
  \end{minipage}
  \hspace{2em}
  \begin{minipage}{0.5\columnwidth}
    \includegraphics[width=1\columnwidth]{figs/three_entropies.tikz}
  \end{minipage}
  \caption{\emph{Top left}: When depicting a measure on the union of two sets $\mu(A \cup B)$, the
    area of each section can be used to represent the inequality \eq{measure_inequalities} and hence
    the values $\mu(A \sm B)$, $\mu(B \sm A)$ and $\mu(A \cap B)$ correspond to the area of each
    section.  This correspondence can be generalised to consider an arbitrary number of sets.
    \emph{Bottom left}: When depicting the joint entropy $H(X,Y)$, the area of each section can also
    be used to represent the inequality \eq{shannon_inequalities} and hence the values $H(X|Y)$,
    $H(Y|X)$ and $I(X;Y)$ correspond to the area of each section.  However, this correspondence does
    not generalise beyond two variables.  \emph{Right}: For example, when considering the entropy of
    three variables, the multivariate mutual information $I(X;Y;Z)$ cannot be accurately represented
    using an area since, as represented by the hatching, it is not non-negative.}
  \label{fig:analogy}
\end{figure}

For any pair of random variables $X$ and $Y$, the entropy $H$ satisfies the inequality
\begin{equation}
  \label{eq:shannon_inequalities}
  H(X) + H(Y) \geq H(X,Y) \geq H(X),\, H(Y) \geq 0.
\end{equation}
From this inequality, it is easy to see that the conditional entropies and mutual information are
non-negative,
\begin{align}
  H(X|Y) &= H(X,Y) - H(Y) \geq 0, \label{eq:cond_entropy_X|Y}\\
  H(Y|X) &= H(X,Y) - H(X) \geq 0, \label{eq:cond_entropy_Y|X}\\
  I(X;Y) &= H(X) + H(Y) - H(X,Y) \geq 0. \label{eq:mutual_info_set} 
\end{align}
For any pair of sets $A$ and $B$, a measure $\mu$ satisfies the inequality
\begin{equation}
  \label{eq:measure_inequalities}
  \mu(A) + \mu(B) \geq \mu(A \cup B) \geq \mu(A),\, \mu(B) \geq 0,
\end{equation}
which follows from the non-negativity of measure on the relative complements and the
intersection,
\begin{alignat}{3}
  \mu(A \sm B) &= \mu(A \cup B) - \mu(B) \geq 0  \label{eq:sm_measure_B}\\
  \mu(B \sm A) &= \mu(A \cup B) - \mu(A) \geq 0  \label{eq:sm_measure_A}\\
  \mu(A \cap B) &= \mu(A) + \mu(B) - \mu(A \cup B) \geq 0 .  \label{eq:intersection_measure}
\end{alignat}
Although the entropy is not itself a measure, several authors have noted the entropy is analogous to
measure in this regard \citep{reza1961, hu1962, abramson1963, campbell1965, csiszar1981, yeung1991,
  yeung2008}.  Indeed, it is this analogy which provides the justification for the typical depiction
of a pair of entropies using Venn diagrams, i.e.\ \fig{analogy}.  Nevertheless, \citet{mackay2003}
notes that this representation is misleading for at least two reasons: Firstly, since the measure on
the intersection $\mu(A \cap B)$ is a measure on a set, it gives the false impression that the
mutual information $I(X;Y)$ is the entropy of some intersection between the random variables.
Secondly, it might lead one to believe that this analogy can be generalised beyond two variables.
However, the analogy does not generalise beyond two variables since the multivariate mutual
information between three random variables,
\begin{align}
  \label{eq:mmi}
  I(X;Y;Z) &= H(X) + H(Y) + H(Z) - H(X,Y) - H(X,Z) - H(Y,Z) + H(Z,Y,Z),
\end{align}
is not non-negative \citep{fano1961, abramson1963}, and hence is not analogous to measure on the
triple intersection $\mu(A \cap B \cap C)$ \citep{abramson1963}.  Indeed, this ``unfortunate''
property led Cover and Thomas to conclude that ``there isn't really a notion of mutual information
common to three random variables''~\citep[p.49]{cover2012}.  Consequently, \citet{mackay2003}
recommends against depicting the entropy of three or more variables using a Venn diagram, i.e.\
\fig{analogy}, unless one is aware of these issues with this representation. 

However, \citet{yeung1991} showed that there is an analogy between entropy and \emph{signed} measure
that is valid for an arbitrary number of random variables.  To do this, Yeung defined a signed
measure on a suitably constructed sigma-field that is uniquely determined by the joint entropies of
the random variables involved. This correspondence enables one to establish information-theoretic
identities from measure-theoretic identities.  Thus, Venn diagrams can be used to represent the
entropy of three or more variables provided one is aware that the certain overlapping areas may
correspond to negative quantities.  Moreover, the multivariate mutual information is useful both as
summary quantity and for manipulating information-theoretic identities provided one is mindful it
may have ``no intuitive meaning''~\citep{yeung1991, csiszar1981}.

In this paper, we introduce new measures of multivariate information that are analogous to measures
upon sets and maintain their operational meaning when considering an arbitrary number of variables.
These new measures complement the existing measures of multivariate mutual information, and will be
constructed by considering the distinct ways in which a set of marginal observers might share their
information with a non-observing third party.  In \secRef{mic}, we will discuss the existing
measures of information content in terms of a set of individuals who each have different knowledge
about a joint realisation from a pair of random variables.  Then in \secRef{marginal}, we will
discuss how these individuals can share their information with a non-observing third party, and
derive the functional form of this individual's information.  In \secRef{synergy}, we relate this
new measure of information content back to the mutual information.
Sections~\ref{sec:properties}--\ref{sec:mid} will then generalise the arguments of
Sections~\ref{sec:marginal} and \ref{sec:synergy} to consider an arbitrary number of observers.
Finally, in \secRef{info}, we will discuss how these new measures can be combined to define new
measures of mutual information.

\section{Mutual Information Content}
\label{sec:mic}

Suppose that Alice and Bob are separately observing some process and let the discrete random
variables $X$ and $Y$ represent their respective observations.  Say that Johnny is a third
individual who can simultaneously make the same observations as Alice and Bob such that his
observations are given by the joint variable $(X,Y)$.  When a realisation $(x,y)$ occurs, Alice's
information is given by the information content~\citep{mackay2003},
\begin{equation}
  \label{eq:info_content}
  h(x) = - \log p_{X}(x) \geq 0,
\end{equation}
where $p_{X}(x)$ is the probability mass of the realisation $x$ of variable $X$ computed from the
probabilty distribution $p_{X}$.  Likewise, Bob's information is given by the information content
$h(y)$, while Johnny's information is by the joint information content $h(x,y)=-\log p_{XY}(x,y)$.
The information that Alice can expect to gain from an observation is given by the entropy,
\begin{equation}
  \label{eq:entropy}
  H(X) = \mathrm{E}_{X} \big[h(x)\big] \geq 0,
\end{equation}
where $\mathrm{E}_{X}$ represents an expectation value over realisations of the variable $X$.
Similarly, Bob's expected information gain is given by the entropy $H(X)$ and Johnny's
expected information is given by the joint entropy $H(X,Y)=\mathrm{E}_{XY} [h(x,y)]$.  Clearly, for
any realisation, Johnny has at least as much information as either Alice or Bob,
\begin{equation}
  \label{eq:shannon_inequalities_pw}
  h(x,y) \geq h(x), h(y) \geq 0.
\end{equation}
The conditional information content can be used to quantify how much more information Johnny has
relative to either Alice or Bob, respectively,
\begin{alignat}{3}
  &h(x|y) &&= h(x,y) - h(y) &&\geq 0, \label{eq:cond_info_content_y}\\
  &h(y|x) &&= h(x,y) - h(x) &&\geq 0. \label{eq:cond_info_content_x}
\end{alignat}
Similarly, we can quantify how much more information Johnny expects to get compared to either Alice
or Bob via the conditional entropies,
\begin{alignat}{3}
  &H(X|Y) &&= \mathrm{E}_{XY} \big[h(x|y)\big] &&\geq 0, \label{eq:cond_entropy_XY} \\
  &H(Y|X) &&= \mathrm{E}_{XY} \big[h(y|x)\big] &&\geq 0. \label{eq:cond_entropy_YX}
\end{alignat}

Now consider a fourth individual who does not directly observe the process, but with whom Alice and
Bob share their knowledge.  To be explicit, we are considering the situation whereby this individual
knows that the joint realisation $(x,y)$ has occurred and knows the marginal distributions $p_X$ and
$p_Y$, but does not know the joint distribution $p_{XY}$.  How much information does this individual
obtain from the shared marginal knowledge provided by Alice and Bob?  The answer to this question
will be provided in \secRef{marginal}, but for now let us consider a simplified version of this
problem.  Suppose that such an individual, whom we will call Indiana, assumes that Alice's
observations are independent of Bob's observations.  In terms of the probabilities, this means that
Indy believes that the joint probability $p_{XY}(x,y)$ is equal to the product probability
$p_{X \times Y}(x,y) = p_X(x)\,p_Y(y)$, while in terms of information, this assumption leads Indiana
to believe that her information is given by the independent information content $h(x) + h(y)$.
Moreover, the information that Indiana expects to gain from any one realisation is given by
$H(X)+H(Y)$.

Let us now compare how much information Indiana believes that she has compared to our other
observers.  For every realisation, Indiana believes that she has at least as much information as
either Alice or Bob,
\begin{equation}
  h(x) + h(y) \geq h(x), h(y) \geq 0.
\end{equation}
Since Indy knows what both Alice and Bob know individually, it is hardly surprising that she always
has at least as much information as either Alice or Bob.  The comparison between Indiana and Johnny,
however, is not so straightforward---there is no inequality that requires the information content of
the joint realisation to be less than the information content of the independent realisations, or
vice versa.  Consequently, the difference between the information that Indiana thinks she has and
Johnny's information, i.e.\ the mutual information content between a pair of realisations,
\begin{equation}
  \label{eq:mutual_info_content}
  i(x;y) = h(x) + h(y) - h(x,y) = \log \frac{p_{XY}(x,y)}{p_X(x)\,p_Y(y)},
\end{equation}
is not non-negative~\citep{fano1959}.  (This function goes by several different names including the
pointwise mutual information, the information density~\citep{pinsker1964} or simply the mutual
information~\citep{fano1961}.)  Thus, similar to how it is potentially misleading to depict the
entropy of three of more variables using a Venn diagram, representing the information content of two
variables using a Venn diagram is somewhat dubious (see \fig{analogy_pw}).

\begin{figure}[t]
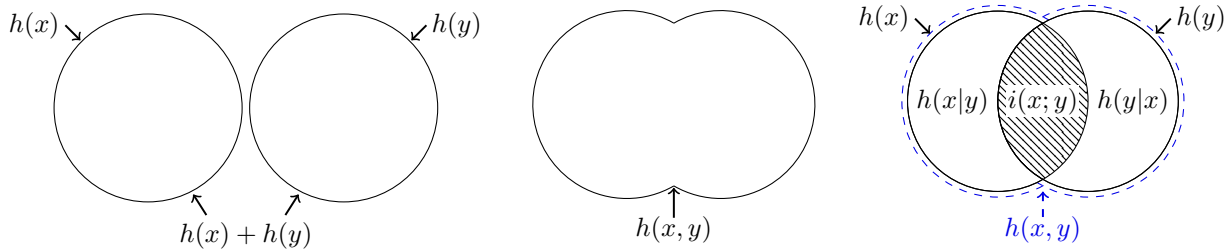

  \begin{minipage}{0.38\textwidth}
    \raggedright
    \vspace{5pt}
    \includegraphics{figs/two_pw_entropies_indep.tikz} \\
  \end{minipage}
  \begin{minipage}{0.25\textwidth}
    \centering
    \includegraphics{figs/two_pw_entropies_joint.tikz} \\
  \end{minipage}
  \begin{minipage}{0.35\textwidth}
    \raggedleft
    \includegraphics{figs/two_pw_entropies.tikz}
  \end{minipage}
  \caption{\emph{Left}: Indiana assumes that Alice's information $h(x)$ is independent of Bob's
    information $h(y)$ such that her information is given by $h(x)+h(y)$.  \emph{Middle}: Johnny
    knows the joint distribution $p_{XY}$, and hence his information is given by the joint
    information content $h(x,y)$.  \emph{Right}: There is no inequality that requires Johnny's
    information to be no greater than Indiana's assumed information, or vice versa.  On one hand,
    Johnny can have more information than Indiana since a joint realisation can be more surprising
    than both of the individual marginal realisations. On the other hand, Indiana can have more
    information than Johnny since a joint realisation can be less surprising than both of the
    individual marginal realisations occurring independently.  Thus, as represented by the hatching,
    the mutual information content $i(x;y)$ is not non-negative.}
  \label{fig:analogy_pw}
\end{figure}

Since Johnny knows the joint distribution $p_{XY}$, while Indiana only knows the marginal
distributions $p_X(x)$ and $p_Y(y)$, we might expect that Indiana should never have more information
than Johnny.  However, Indiana's assumed information is based upon the belief that Alice's
observations $X$ are independent of Bob's observations $Y$, which leads Indiana to overestimate her
information on average.  Indeed, Indiana is so optimistic that the information she expects to get upper
bounds the information that Johnny can expect to get,
\begin{equation}
  \label{eq:half_ineq}
  H(X) + H(Y) \geq H(X,Y) \geq 0.
\end{equation}
Thus, despite the fact that Indiana can have less information than Johnny for certain
realisations---i.e.\ despite the fact that the mutual information content is not non-negative---the
mutual information in expectation is non-negative,
\begin{equation}
  \label{eq:mutual_info}
  I(X;Y) =  H(X) + H(Y) - H(X,Y) = \mathrm{E}_{XY} \big[i(x;y)\big] \geq 0.
\end{equation}
Crucially, and in contrast to the information content \eq{info_content} and entropy \eq{entropy},
the non-negativity of the mutual information does not follow directly from the non-negativity of the
mutual information content \eq{mutual_info_content}, but rather must be proved separately.
(Typically, this is done by showing that the mutual information can be written as a Kullback-Leibler
divergence which is non-negative by Jensen's inequality, e.g.\ see \citet{cover2012}.)  Thus, not
only does Indiana potentially have more information than Johnny for certain realisations, but on
average we expect Indiana to have more information than Johnny.  Of course, by assuming Alice's
observations are independent of Bob's observations, Indiana is overestimating her information.
Thus, in the next section, we will consider the situation whereby one does not make this assumption.

\section{Marginal Information Sharing}
\label{sec:marginal}

Suppose that Eve is another individual who, similar to Indiana, does not make any direct
observations, but with whom both Alice and Bob share their knowledge; i.e.\ Eve knows the joint
realisation $(x,y)$ has occurred and knows the marginal distributions $p_X$ and $p_Y$, but does not
know the joint distribution~$p_{XY}$.  Furthermore, suppose that Eve is more conservative than
Indiana and does not assume that Alice's observations are independent of Bob's observations---how
much information does Eve have for any one realisation?

It seems clear that Eve's information should always satisfy the following two requirements. Firstly,
since Alice and Bob both share their knowledge with Eve, she should have at least as much
information as either of them have individually.  Secondly, since Eve has less knowledge than
Johnny, she should have no more information than Johnny; i.e.\ in contrast to Indy, Eve should never
have more information than Johnny.  As the following theorem shows, these two requirements uniquely
determine the functional form of Eve's information:
\begin{theorem}
  \label{thm:eve}
  The unique function $h(x \sqcup y)$ of $p_X(x)$ and $p_Y(y)$ that satisfies
  $h(x,y) \geq h(x \sqcup y) \geq h(x),\,h(y) \geq 0$ for all $p_{XY}(x,y)$ is
  \begin{equation}
    \label{eq:union_info_content}
    h(x \sqcup y) = \max\big(h(x),h(y)\big) \geq 0.
  \end{equation}
\end{theorem}
\begin{proof}
  Clearly, the function is lower bounded by $\max\big(h(x),h(y)\big)$.  The upper bound is given by
  the minimum possible $h(x,y)$ which corresponds to the maximum allowed $p_{XY}(x,y)$.  For any
  $p_X(x)$ and $p_Y(y)$, the maximum allowed $p_{XY}(x,y)$ is $\min\big(p_X(x),p_Y(y)\big)$, which
  corresponds to $h(x,y)=\max\big(h(x),h(y)\big)$.
\end{proof}
Eve's information is given by the maximum of Alice's and Bob's information, or the information
content of the most surprising marginal realisation.  Although we have defined Eve's information by
requiring it to be no greater than Johnny's information, it is also clear that Eve also has no more
information than Indiana.  As such, Eve's information satisfies the inequality
\begin{equation}
  \label{eq:union_inequality}
  h(x) + h(y) \geq h(x \sqcup y) \geq h(x),\,h(y) \geq 0,
\end{equation}
which is analogous to the inequality \eq{measure_inequalities} satisfied by measure.  Hence, as
pre-empted by the notation (and will be further justified in \secRef{generalised}), Eve's
information will be referred to as the \emph{union information content}.  The union information
content is the maximum possible information that Eve can get from knowing what Alice and Bob
know---it quantifies the information provided by a joint event $(x,y)$ when one knows the marginal
distributions $p_X$ and $p_Y$, but does not know nor make any assumptions about the joint
distribution~$p_{XY}$.

Similar to how the conditional information contents \eq{cond_entropy_XY} and \eq{cond_entropy_YX}
enables us to quantify how much more information Johnny has relative to either Alice or Bob, the
inequality \eq{union_inequality} enables us to quantify how much information Eve gets from Alice
relative to Bob and vice versa, respectively,
\begin{alignat}{4}
  &h(x \ssm y) &&= h(x \sqcup y) - h(y) &&= \max\big(h(x) - h(y),0\big) &&\geq 0,\label{eq:uni_ic_x}\\
  &h(y \ssm x) &&= h(x \sqcup y) - h(x) &&= \max\big(0,h(y) - h(x)\big) &&\geq 0.\label{eq:uni_ic_y}
\end{alignat}
These non-negative functions are analogous to measure on the relative complements of a pair of sets
and will be called the \emph{unique information content from $x$ relative to $y$}, and vice versa
respectively.  It is easy to see that, since Eve's information is either equal to Alice's or Bob's
information (or both), at least one of these two functions must be equal to zero.  

The inequality \eq{union_inequality} also enables us to quantify how much more information Indiana
has relative to Eve.  Since Indiana's assumed information is given by the sum of Alice's and Bob's
information while Eve's information is given by the maximum of Alice's and Bob's information, the
difference between the two is given by the minimum of Alice's and Bob's information,
\begin{equation}
  \label{eq:max_min_two}
  h(x \sqcap y) = h(x) + h(y) - h(x \sqcup y) = h(x) + h(y) - \max\big(h(x), h(y) \big) 
  = \min\big(h(x), h(y)\big) \geq 0. 
\end{equation}
In contrast to the comparison between Indiana and Johnny, i.e.\ the mutual information content
\eq{mutual_info_content}, the comparison between Indiana and Eve is non-negative.  As such, this
function is analogous to measure on the intersection of two sets and hence will be referred to as
the \emph{intersection information content}.  The intersection information content is the minimum
possible information that Eve could have gotten from knowing either what Alice or Bob know, and is
given by the information content of the least surprising marginal realisation.

Finally, from \eq{union_info_content} and \eq{uni_ic_x}--\eq{max_min_two}, it is not difficult to
see that Eve's information can be decomposed into the information that could have been obtained from
either Alice or Bob, the unique information from Alice relative to Bob and the unique information
from Bob relative to Alice,
\begin{equation}
  \label{eq:decomp_union}
  h(x \sqcup y) = h(x \sqcap y) + h(x \ssm y) + h(y \ssm x).
\end{equation}
Of course, as already discussed, at least one of these unique information contents must be zero.
\fig{union_info} depicts this decomposition for some realisation whereby Alice's information $h(x)$
is greater than Bob's information $h(y)$.

To summarise thus far, both Alice and Bob share their information with Indiana and Eve, who then
each interpret this information in a different way.  By comparing \fig{analogy_pw} and
\fig{union_info}, we can easily contrast their distinct perspectives.  Eve is more conservative than
Indiana and assumes that she has gotten as little information as she could possibly have gotten from
knowing what Alice and Bob know; this is given by the maximum from Alice's and Bob's information, or
is the information content associated with the most surprising marginal realisation observed by
Alice and Bob.  In effect, Eve's conservative approach means that she pessimistically assumes that
the information provided by the least surprising marginal realisation was already provided by the
most surprising marginal realisation.  In contrast, Indiana optimistically assumes that the
information provided by the least surprising marginal realisation is independent of the information
provided by the most surprising marginal realisation.


\begin{figure}[t]
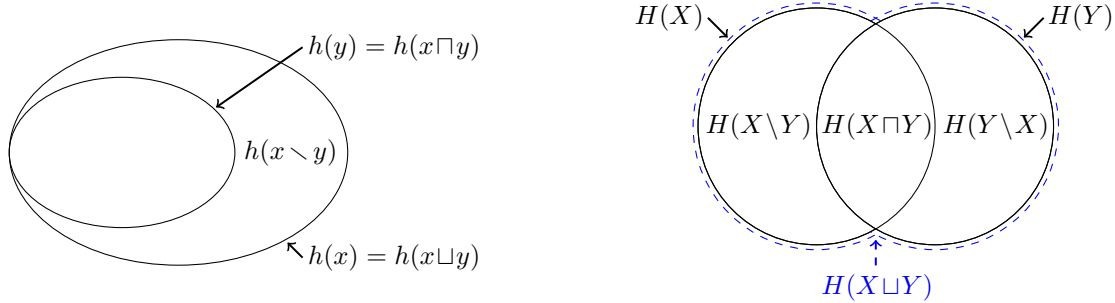

  \begin{minipage}{0.5\textwidth}
    \centering
    \includegraphics{figs/union_entropy_two.tikz}
  \end{minipage}
  \begin{minipage}{0.5\textwidth}
    \centering
    \includegraphics[width=0.8\columnwidth]{figs/avg_union_entropy.tikz}
  \end{minipage}
  \caption{\emph{Left}: If Alice's information $h(x)$ is greater than Bob's information $h(y)$, then
    Eve's information $h(x \sqcup y)$ is equal to Alice's information $h(x)$.  In effect, Eve is
    pessimistically assuming that information provided by the least surprising marginal realisation
    $h(x \sqcap y)$ is already provided by the most surprising marginal realisation $h(x \sqcup y)$,
    i.e.\ Bob's information $h(y)$ is a subset of Alice's information $h(x)$.  From this
    perspective, Eve gets unique information from Alice relative to Bob $h(x \ssm y)$, but does not
    get any unique information from Bob relative to Alice $h(y \ssm x)=0$.  \emph{Right}: Although
    for each realisation Eve can only get unique information from either Alice or Bob, it is
    possible that Eve can expect to get unique information from both Alice and Bob on average.  (Do
    not confuse this representation of the union entropy with the diagram that represents the joint
    entropy in \fig{analogy}.)}
  \label{fig:union_info}
\end{figure}

Let us now consider the information that Eve expects to get from a single realisation,
\begin{equation}
  \label{eq:union_entropy}
  H(X \sqcup Y) = \mathrm{E}_{XY} \big[h(x \sqcup y)\big] \geq 0.
\end{equation}
This function will be called the union entropy, and quantifies the expected surprise of the most
surprising realisation from either $X$ or $Y$.  Similar to how the non-negativity of the entropy
\eq{entropy} follows from the non-negativity of the information content \eq{info_content}, the
non-negativity of the union entropy \eq{union_entropy} follows directly from the non-negativity of
the union information content \eq{union_info_content}.  Since the expectation value is monotonic,
and since the union information content satisfies the inequality \eq{union_inequality}, we get that
the union entropy satisfies
\begin{equation}
  \label{eq:union_entropy_inequality}
  H(X) + H(Y) \geq H(X \sqcup Y) \geq H(X),\,H(Y) \geq 0,
\end{equation}
and hence is also analogous to measure on the union of two sets.

Using this inequality, we can quantify how much more information Eve expects to get from Alice
relative to Bob, or vice versa respectively,
\begin{alignat}{4}
  &H(X \sm Y)&&=H(X \sqcup Y) - H(Y)&&=\mathrm{E}_{XY} \big[h(x \ssm y)]&&\geq 0,\label{eq:uni_e_x}\\
  &H(Y \sm X)&&=H(X \sqcup Y) - H(X)&&=\mathrm{E}_{XY} \big[h(y \ssm x)]&&\geq 0,\label{eq:uni_e_y}
\end{alignat}
These functions are also analogous to measure on the relative complements of a pair of sets and
hence will be called the unique entropy from $X$ relative to $Y$, and vice versa respectively.
Crucially, and in contrast to \eq{uni_ic_x} and \eq{uni_ic_y}, both of these quantities can be
simultaneously non-zero; although Alice might observe the most surprising event in one joint
realisation, Bob might observe the most surprising event in another and hence both functions can be
simultaneously non-zero.

Now consider how much more information Indiana expects to get relative to Eve,
\begin{equation}
  \label{eq:intersection_entropy}
  H(X \sqcap Y) = H(X) + H(Y) - H(X \sqcup Y) = \mathrm{E}_{XY} \big[h(x \sqcap y)\big] \geq 0.
\end{equation}
This function is also analogous to measure on the intersection of two sets function will be called
the intersection entropy.  In contrast to the mutual information \eq{mutual_info}, since the
intersection information content \eq{max_min_two} is non-negative, we do not require an additional
proof to show that the intersection entropy is non-negative.

Finally, similar to \eq{decomp_union}, we can decompose Eve's expected information into the
following components,
\begin{equation}
  \label{eq:decomp_union_entropy}
  H(X \sqcup Y) = H(X \sqcap Y) + H(X \sm Y) + H(Y \sm X).
\end{equation}
It is important to reiterate that, in contrast to \eq{decomp_union}, there is nothing which requires
either of the two unique entropies to be zero.  Thus, as shown in \fig{union_info}, the Venn diagram
which represents the union and intersection entropy differs from that which represents the union
information content.

\section{Synergistic Information Content}
\label{sec:synergy}

\begin{figure}[t]
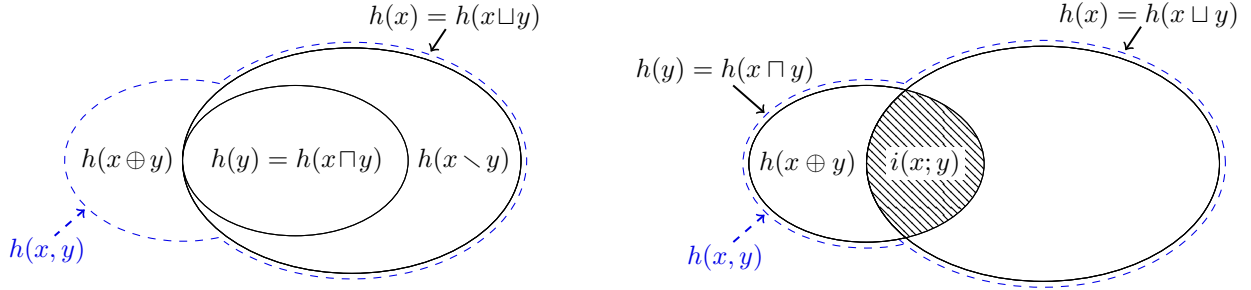

  \begin{minipage}{0.5\textwidth}
    \raggedright
    \includegraphics{figs/decomp_entropy_two.tikz}
  \end{minipage}
  \begin{minipage}{0.5\textwidth}
    \raggedleft
    \includegraphics[width
    =1\columnwidth]{figs/synergistic_entropy_two.tikz}
  \end{minipage}
  \caption{\emph{Left}: This Venn diagram shows how the synergistic information $h(x \oplus y)$ can
    be defined by comparing the joint information content $h(x,y)$ from \fig{analogy_pw} to the
    union information content $h(x \sqcup y)$ from \fig{union_info}.  Note that, for this particular
    realisation, we are assuming that $h(x) > h(y)$.  It also provides a visual representation of
    the decomposition \eq{bivar_decomp_avg} of the joint information content $h(x,y)$.
    \emph{Right}: By rearranging the marginal entropies such that they match \fig{analogy_pw}
    (albeit with different sizes here), it is easy to see why the mutual information content
    $i(x;y)$ is equal to the intersection information content $h(x \sqcap y)$ minus the synergistic
    information content $h(x \oplus y)$, c.f.\ \eq{mi_diff}.}
  \label{fig:synergistic_info_content}
\end{figure}

As we discussed at the beginning of the previous section, and as we required in
Theorem~\ref{thm:eve}, one of the defining features of Eve's information is that it is never greater
than Johnny's information,
\begin{equation}
  \label{eq:joint_inequality}
  h(x,y) \geq h(x \sqcup y).
\end{equation}
Thus, we can compare how much more information Johnny has relative to Eve,
\begin{equation}
  \label{eq:synergistic_ic}
  h(x \oplus y) = h(x,y) - h(x \sqcup y)
    = h(x,y) - \max\big(h(x), h(y) \big)
    = \min\big(h(y|x), h(x|y)\big) \geq 0.
\end{equation}
This non-negative function will be called the \emph{synergistic information content}, and it
quantifies how much more information one gets from knowing the joint probability $p_{XY}(x,y)$
relative to merely knowing the marginal probabilities $p_X(x)$ and $p_Y(y)$.
\fig{synergistic_info_content} shows how this relationship can represented using a Venn diagram.  Of
course, by this definition, Johnny's information is equal to the union information content plus the
synergistic information content, and hence, by using \eq{decomp_union}, we can decompose Johnny's
information into the intersection information content, the unique information contents, and the
synergistic information contents,
\begin{equation}
  \label{eq:bivar_decomp}
  h(x,y) = h(x \sqcup y) + h(x \oplus y) = h(x \sqcap y) + h(x \ssm y) + h(y \ssm x) + h(x \oplus y).
\end{equation}
This decomposition can be seen in \fig{synergistic_info_content}, although it is important to recall
that at least one of $h(x \ssm y)$ and $h(y \ssm x)$ must be equal to zero.  In a similar manner,
the extra information that Johnny has relative to Bob \eq{cond_info_content_y} can be decomposed
into the unique information content from Alice and the synergistic information content, and vice
versa for the extra information that Johnny has relative to Alice \eq{cond_info_content_x},
\begin{align}
  h(x|y) &= h(x \ssm y) + h(x \oplus y), \label{eq:cond_decomp_x} \\
  h(y|x) &= h(y \ssm x) + h(x \oplus y). \label{eq:cond_decomp_y}
\end{align}

Now recall that the mutual information content \eq{mutual_info_content} is given by Indiana's
information minus Johnny's information.  By replacing Johnny's information with the union
information content plus the synergistic information content via \eq{synergistic_ic} and rearranging
using \eq{max_min_two}, we get that the mutual information content is equal to the intersection
information content minus the synergistic information content,
\begin{align}
  \label{eq:mi_diff}
  i(x;y) = h(x) + h(y) - h(x,y) = h(x) + h(y) - h(x \sqcup y) - h(x \oplus y)
    =  h(x \sqcap y) - h(x \oplus y).
\end{align}
Indeed, this relationship can be identified in \fig{synergistic_info_content}.  Clearly, the mutual
information content is negative whenever the synergistic information content is greater than the
intersection information content.  From this perspective, the mutual information content can be
negative because there is nothing to suggest that the synergistic information content should be no
greater than the intersection information content.  In other words, the additional surprise
associated with knowing $p_{XY}(x,y)$ relative to merely knowing $p_X(x)$ and $p_Y(y)$ can exceed
the surprise of the least surprising marginal realisation.

Let us now quantify how much more information Johnny expects to get relative to Eve,
\begin{equation}
  \label{eq:synergistic_e}
  H(X \oplus Y) = \mathrm{E}_{XY} \big[h(x \oplus y)\big] = H(X,Y) - H(X \sqcup Y) \geq 0,
\end{equation}
which we will call the synergistic entropy.  Crucially, although the synergistic information content
is given by the minimum of the two conditional information contents, the synergistic entropy does
not in general equal one of the two the conditional entropies.  This is because, although Alice
might observe the most surprising event in one joint realisation such that the synergistic
information content is equal to Bob's information given Alice's information, Bob might observe the
most surprising event in another realisation such that the synergistic information content is equal
to Alice's information given Bob's information for that particular realisation.  Thus, the
synergistic entropy does not equal the conditional entropy for the same reason that unique entropies
\eq{uni_e_x} and \eq{uni_e_y} can be simultaneously non-zero.

\begin{figure}[t]
  \centering
  \includegraphics[width =0.5\textwidth]{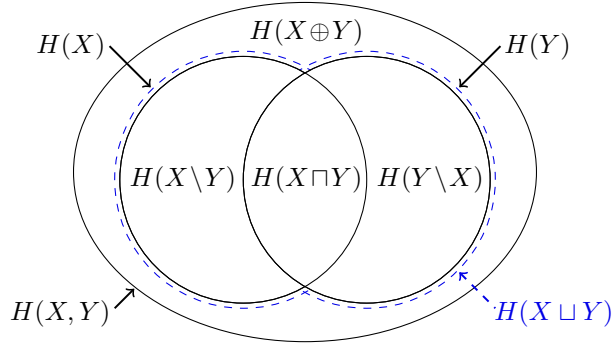}
  \caption{This Venn diagram shows how the synergistic entropy $H(X \oplus Y)$ can be defined by
    comparing the joint entropy $H(X,Y)$ from \fig{analogy} to the union entropy $H(X \sqcup Y)$ from
    \fig{union_info}.  It also provides a visual representation of the decomposition
    \eq{bivar_decomp_avg} of the joint entropy $H(X,Y)$.}
  \label{fig:synergistic_entropy}
\end{figure}

With the definition of synergistic entropy, it is not difficult to show that, similar to
\eq{bivar_decomp}, the joint entropy can be decomposed into the following components,
\begin{equation}
  \label{eq:bivar_decomp_avg}
  H(X,Y) = H(X \sqcup Y) + H(X \oplus Y) = H(X \sqcap Y) + H(X \sm Y) + H(Y \sm X) + H(X \oplus Y).
\end{equation}
\fig{synergistic_entropy} depicts this decomposition using a Venn diagram, and shows how the union
entropy from \fig{union_info} is related to the joint entropy $H(X,Y)$.  Likewise, similar to
\eq{cond_decomp_x} and \eq{cond_decomp_y}, it is easy to see that conditional entropies can be
decomposed as follows,
\begin{align}
  H(X|Y) &= H(X \sm Y) + H(X \oplus Y), \\
  H(Y|X) &= H(Y \sm X) + H(X \oplus Y) 
\end{align}
Finally, just like \eq{mi_diff}, we can also show that the mutual information is equal to the
intersection entropy minus the synergistic entropy,
\begin{equation}
  I(X;Y) =  H(X \sqcap Y) - H(X \oplus Y) \geq 0.
\end{equation}
Although there is nothing to suggest that the synergistic information content must be no greater
than the intersection information content, we know that the synergistic entropy must be no greater
than the intersection entropy because $I(X;Y) \geq 0$.  In other words, the expected difference
between the surprise of the joint realisation and the most surprising marginal realisation cannot
exceed the expected surprise of the least surprising realisation.

\section{Properties of the Union and Intersection Information Content}
\label{sec:properties}

Theorem~\ref{thm:eve} determined the function form of Eve's information when Alice and Bob share
their knowledge with her.  We now wish to generalise this result to consider the situation whereby
an arbitrary number of marginal observers share their information with Eve.  Rather than try to
directly determine the functional form, however, we will proceed by considering the algebraic
structure of shared marginal information.  

If Alice and Bob observe the same realisation $x$ such that they have the same information $h(x)$,
then upon sharing we would intuitively expect Eve to have the same information $h(x)$.  Similarly,
the minimum information that Eve could have received from either Alice or Bob should be the same
information $h(x)$.  Since the maximum and minimum operators are idempotent, the union and
intersection information content both align with this intuition.
\begin{prop}[Idempotence]
  \label{prop: idempotent}
  The union and intersection information content are idempotent,
  \begin{align}
    h(x \sqcup x) &= h(x), \\
    h(x \sqcap x) &= h(x). \label{eq:intersectionIdempotent}
  \end{align}
\end{prop}

It also seems reasonable to expect that Eve's information should not depend on the order in which
Alice and Bob share their information, and nor should the minimum information that Eve could have
received from either individual.  Again, since the maximum and minimum operators are commutative,
the union and intersection information content both align with our intuition.
\begin{prop}[Commutativity]
  \label{prop: commutivity}
  The union and intersection information content are commutative,
  \begin{align}
    h(x \sqcup y) &= h(y \sqcup x), \\
    h(x \sqcap y) &= h(y \sqcap x).
  \end{align}
\end{prop}

Now suppose that Charlie is another individual who, just like Alice and Bob, is separately observing
some process, and let the random variable $Z$ represent her observations.  Say that Dan is yet
another individual with whom, just like Eve, our observers can share their information. Intuitively,
it should not matter whether Alice, Bob and Charlie share their information directly with Eve, or
whether they share their information through Dan.  To be specific, Alice and Bob could share their
information with Dan such that his information is given by $h(x \sqcup y)$, and then Charlie and Dan
could subsequently share their information with Eve such that her information is given by
$h\big((x \sqcup y) \sqcup z \big)$.  Similarly, Bob and Charlie could share their information with
Dan such that his information is given by $h(y \sqcup z)$, and then Alice and Dan could subsequently
share their information with Eve such that her information is given by
$h\big(x \sqcup (y \sqcup z)\big)$.  Alternatively, Alice, Bob and Charlie could entirely bypass Dan
and share their information directly with Eve such that her information is given by
$h(x \sqcup y \sqcup z)$.  Since the maximum operator is associative, the union information content
is the same in all three cases and hence aligns with our intuition.  A similar argument can be made
to show that the intersection information content is also associative.
\begin{prop}[Associative]
  \label{prop: associativity}
  The union and intersection information content are associative,
  \begin{alignat}{3}
    &h(x \sqcup y \sqcup z) &= h\big((x \sqcup y) \sqcup z \big) &= h\big(x \sqcup (y \sqcup z)\big), \\
    &h(x \sqcap y \sqcap z) &= h\big((x \sqcap y) \sqcap z \big) &= h\big(x \sqcap (y \sqcap z)\big),  
  \end{alignat}
\end{prop}

Suppose now that Alice and Bob share their information with Dan such the information that he could
have gotten from either Alice or Bob is given by $h(x \sqcap y)$.  If Alice and Dan both share their
information with Eve, then Eve's information is given by
\begin{equation}
  h\big(x \sqcup (x \sqcap y)\big) = \max\big(h(x),\min\big(h(x),h(y)\big)\big) = h(x),
\end{equation}
and hence Bob's information has been absorbed by Alice's information.  Now suppose that Alice and
Bob share their information with Dan such his information is given by $h(x \sqcup y)$.  If Alice and
Dan both share their information with Eve, then the information that Eve could have gotten from
either Alice or Dan is given by
\begin{equation}
  h\big(x \sqcap (x \sqcup y)\big) = \min\big(h(x),\max\big(h(x),h(y)\big)\big) = h(x).
\end{equation}
Again, Bob's information has been absorbed by Alice's information.  Both of these results are a
consequence of the fact that the maximum and minimum operators are connected to each other by the
absorption identity.
\begin{prop}[Absorption]
  \label{prop: absorption}
  The union and intersection information content are connected by absorption,
  \begin{align}
    h\big(x \sqcup (x \sqcap y)\big) &= h(x), \\
    h\big(x \sqcap (x \sqcup y)\big) &= h(x). 
  \end{align}
\end{prop}

Now say that Daniella is, just like Eve or Dan, an individual with whom our observers can share
their information.  Consider the following two cases: Firstly, suppose that Bob and Charlie share
their information with Dan such that the information that Dan could have gotten from either Bob or
Charlie is given by $h(y \sqcap z)$.  If both Alice and Dan share their information with Eve, then
her information is given by $h\big(x \sqcup (y \sqcap z)\big)$.  In the second case, suppose that
Alice and Bob share their information with Dan such that his information is given by
$h(x \sqcup y)$, while Alice and Charlie simultaneously share their information with Daniella such
that her information is given by $h(x \sqcup z)$.  If Dan and Daniella both share their information
with Eve, then the information that she could have gotten from either Dan or Daniella is then given
by $h\big((x \sqcup y) \sqcap (x \sqcup z)\big)$.  In both cases, Eve has the same information since
the maximum operator is distributive,
\begin{align}
  h\big(x \sqcup (y \sqcap z)\big) &= \max\big(h(x),\min\big(h(y),h(z)\big)\big) \nn\\
    &=\min\big(\max\big(h(x),h(y)\big),\max\big(h(x),h(z)\big)\big)
      = h\big((x \sqcup y) \sqcap (x \sqcup z)\big).
\end{align}
Since the maximum and minimum operators are distributive over each other, regardless of whether Eve
gets Alice's information and Bob's or Charlie's information, or if Eve gets Alice's and Bob's
information or Alice's and Charlie's information, Eve has the same information. The same reasoning
can be applied to show that regardless of whether Eve gets Alice's information or Bob's and
Charlie's information, or if Eve gets Alice's or Bob's information and Alice's or Charlie's
information, Eve has the same information.
\begin{prop}[Distributivity]
  \label{prop: distributivity}
  The union and intersection information content are distribute over each other,
  \begin{align}
    h\big(x \sqcup (y \sqcap z)\big) &= h\big((x \sqcup y) \sqcap (x \sqcup z)\big),
                                   \label{eq:union_distrib}\\
    h\big(x \sqcap (y \sqcup z)\big) &= h\big((x \sqcap y) \sqcup (x \sqcap z)\big).
  \end{align}
\end{prop}

\begin{figure}[p]
  \includegraphics{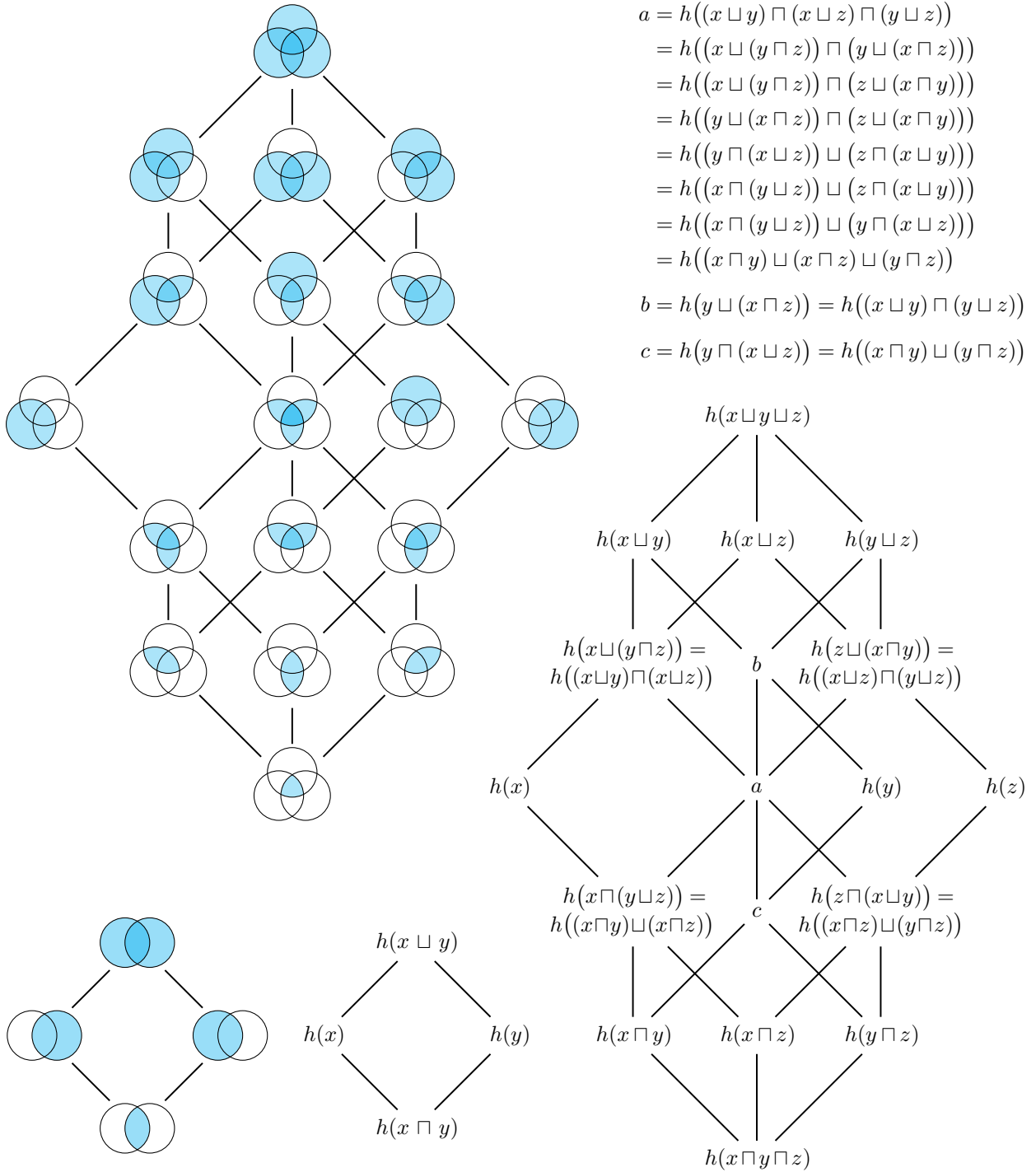}
  \caption{\emph{Bottom right}: The distributive lattices
    $\langle\bm{x},h(\,\sqcup\,),h(\,\sqcap\,)\rangle$ of information contents for two and three and
    three observers.  It is also important to note that, by replacing $h$, $x$, $y$ and $z$ with
    $H$, $X$, $Y$ and $Z$ respectively, we can obtain the distributive lattices for entropy.  In
    fact, this is crucial since Property~\ref{prop:connex} enables us to reduce the distributive
    lattice of information contents to a mere total order; however, this property does not apply to
    the entropies, and hence we cannot further simplify the lattice of entropies.  \emph{Top left}:
    By the fundamental theorem of distributive lattices, the distributive lattices of marginal
    information contents has a one-to-one correspondence with the lattice of sets.  Notice that the
    lattice for two sets corresponds to the Venn diagram for entropies in \fig{union_info}.}
  \label{fig:distributive_lattices}
\end{figure}

Now consider a set of $n$ individuals and let $\bm{X} = \{ X_1, X_2, \ldots, X_n\}$ be the joint
random variable that represents their observations.  Suppose that these individuals together observe
the joint realisation $\bm{x} = \{ x_1, x_2, \ldots, x_n\}$ from $\bm{X}$.  By Property~\ref{prop:
  associativity} and the general associativity theorem, it is clear that Eve's information is given
by
\begin{equation}
  h(x_1 \sqcup x_2 \sqcup \ldots \sqcup x_n ) = \max\big(h(x_1), h(x_2), \ldots, h(x_n) \big) \geq 0,
\end{equation}
while the minimum information that Eve could have gotten from any individual observer is given by
\begin{equation}
  h(x_1 \sqcap x_2 \sqcap \ldots \sqcap x_n ) = \min\big(h(x_1), h(x_2), \ldots, h(x_n) \big) \geq 0. 
\end{equation}
This accounts for the situation whereby $n$ marginal observers directly share their information with
Eve, and could clearly be considered for any subset $\bm{S}$ of the observers $\bm{X}$.  We now wish
to consider all of the distinct ways that these marginal observers can share their information
indirectly with Eve.  As the following theorem shows, Properties~\ref{prop: idempotent}--\ref{prop:
  distributivity} completely characterise the unique methods of marginal information sharing.

\begin{theorem}
  The marginal information contents form a join semi-lattices $\langle\bm{x},h(\,\sqcup\,)\rangle$
  under the $\max$ operator.  Separately, the marginal information contents form a meet semi-lattice
  $\langle\bm{x},h(\,\sqcap\,)\rangle$ under the $\min$ operator.
\end{theorem}
\begin{proof}
  Properties~\ref{prop: idempotent}--\ref{prop: associativity} completely characterise semi-lattices
  \citep{gratzer2002, davey2002}.
\end{proof}

\begin{theorem}
  The marginal information contents form a distributive lattice
  $\langle\bm{x},h(\,\sqcup\,),h(\,\sqcap\,)\rangle$ under the $\max$ and $\min$ operators.
\end{theorem}
\begin{proof}
  From Property~\ref{prop: absorption}, we have that the semi-lattices
  $\langle\bm{x},h(\,\sqcup\,)\rangle$ and $\langle\bm{x},h(\,\sqcap\,)\rangle$ are connected by
  absorption and hence form a lattice $\langle\bm{x},h(\,\sqcup\,),h(\,\sqcap\,)\rangle$.  By
  Property~\ref{prop: distributivity} this is a distributive lattice \citep{gratzer2002, davey2002}.
\end{proof}

Each way that a set of $n$ observers can share their information with Eve such that she has distinct
information corresponds to an element in partially ordered set, or more specifically the free
distributive lattice on $n$ generators \citep{gratzer2002}.  \fig{distributive_lattices} shows the
free distributive lattices generated by $n=2$ and $n=3$ observers.  The number of elements in this
lattice given by the ($n$)-th Dedekind number \citep[p.273]{comtet2012} (see also \citep{oeis}).
By the fundamental theorem of distributive lattices (or Birkhoff's representation theory), there is
isomorphism between the union information content and set union, and between the intersection
information content and set intersection \citep{birkhoff1940, stanley1997, gratzer2002, davey2002}.
It is this one-to-one correspondence that justifies our use of the terms union and intersection
information content for $n$ variables in general.  Every identity that holds in a lattice of sets
will have a corresponding identity in this distributive lattice of information contents.
\fig{distributive_lattices} also depicts the sets which correspond to each term in the lattice of
information contents. Just as the cardinality of sets is non-decreasing as we consider moving up
through the various terms in a lattice of sets, Eve's information is non-decreasing as we moving up
through the various terms in the corresponding lattice of information contents.  In particular, we
can quantify the unique information content that Eve gets from one method of information sharing
relative to any other method that is lower in the lattice.

Every property of the union and intersection information content that we have considered thus far
has been directly inherited by the union and intersection entropy.  However, there is one final
property is not inherieted by the entropies.  If Alice and Bob share their information with Eve,
then Eve's information is given by either Alice's or Bob's information, and similar for the
information that Eve could have gotten from either Alice or Bob.  As the subsequent theorem shows,
this property enables us to greatly reduce the number of distinct terms in the distributive lattice
for information content since any partially ordered set with a connex relation forms a total order.

\begin{prop}[Connexity]
  \label{prop:connex}
  The union and intersection information content are given by at least one of
  \begin{equation}
    h(x \sqcup y) = h(x) \; \text{and} \enspace h(x \sqcap y) = h(y), \quad\text{or}\quad
    h(x \sqcup y) = h(y) \; \text{and} \enspace h(x \sqcap y) = h(x) 
  \end{equation}
\end{prop}

\section{Generalised Marginal Information Sharing}
\label{sec:generalised}

We will now use Properties~\ref{prop: idempotent}--\ref{prop:connex} to generalise the results of
Theorem~\ref{thm:eve} and \secRef{marginal}.

\begin{theorem}
  The marginal information contents are a totally ordered set under the $\max$ and $\min$ operators.
\end{theorem}
\begin{proof}
  A totally ordered set is a partially ordered set with the connex
  property~\citep[p.2]{gratzer2002}.
\end{proof}

\begin{figure}[t]
  \centering
  \includegraphics{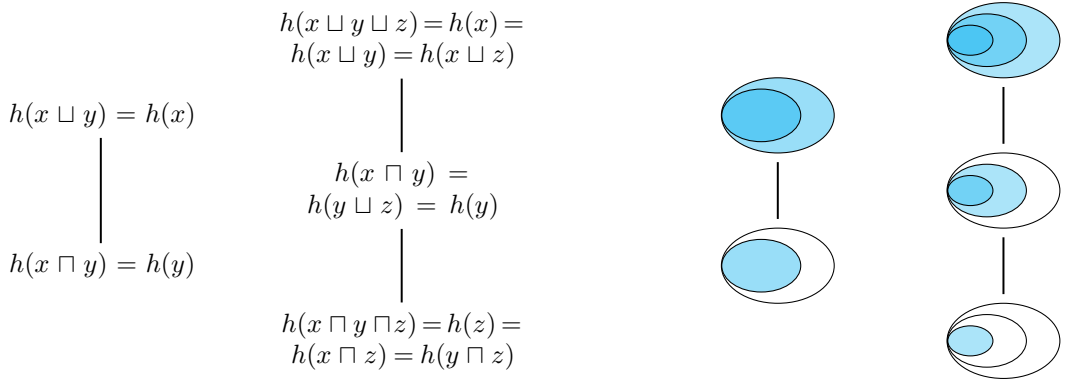}
  \caption{\emph{Left}: The total order of marginal information contents for two and three
    observers, whereby we have assumed that Alice's information $h(x)$ is greater than Bob's
    information $h(y)$, which is greater than Charlie's information $h(z)$.  It is important to note
    that taking the expectation value over these information contents for each realisation, which
    may each have a different total orders, yields entropies which are merely partially ordered.  It
    is for this reason that Property~\ref{prop:connex} does not apply to entropies.  \emph{Right}:
    The Venn diagrams corresponding to the total order for for two and three observers and their
    corresponding information contents.  Notice that the total order for two sets corresponds to the
    Venn diagram for information contents in \fig{union_info}.}
  \label{fig:total_order}
\end{figure}

\fig{total_order} shows the totally ordered sets generated by $n=2$ and $n=3$ observers, and also
depicts the corresponding sets.  Although the number of distinct terms has been reduced, Eve's
information is still non-decreasing as we move up through terms of the totally ordered set.  If we
now compare how much unique information Eve gets from a given method of information sharing relative
to any other method of information sharing which is equal or lower in the totally ordered set, then
we obtain a result which generalises \eq{uni_ic_x} and \eq{uni_ic_y} to consider more than two
observers.  Similarly, this total order enables us to generalise \eq{max_min_two} using the
maximum-minimum identity~\citep{sheldon2002}, which is a form of the principle of
inclusion-exclusion \citep{stanley1997} for a totally ordered set,
\begin{align}
  h(x_1 \sqcap x_2 \sqcap \ldots \sqcap x_n) 
  &= \min\big(h(x_1), h(x_2), \ldots, h(x_n)\big) \nn\\
  &= \sum_{k=1}^n (-1)^{k-1} \sum_{\substack{\bm{S}\subseteq\bm{X}\\|\bm{S}|=k }}
    \max\big(h(s_1), h(s_2), \ldots, h(s_k)\big) \nn\\
  &= \sum_{k=1}^n (-1)^{k-1} \sum_{\substack{\bm{S}\subseteq\bm{X}\\|\bm{S}|=k }}
    h(s_1 \sqcup s_2 \sqcup \ldots \sqcup s_k),
\end{align}
or conversely,
\begin{equation}
  \label{eq:pie_union}
  h(x_1 \sqcup x_2 \sqcup \ldots \sqcup x_n)
    = \sum_{k=1}^n (-1)^{k-1} \sum_{\substack{\bm{S}\subseteq\bm{X} \\
    |\bm{S}|=k }} h(s_1 \sqcap s_2 \sqcap \ldots \sqcap s_k). 
\end{equation}

Now that we have generalised the union and intersection information content, similar to
\secRef{marginal}, let us now consider taking the expectation value for each term in the
distributive lattice.  For every joint realisation $\bm{x}$ from $\bm{X}$, there is a corresponding
distributive lattice of information contents.  Hence, similar to \eq{union_entropy} and
\eq{uni_e_x}--\eq{intersection_entropy}, we can consider taking the expectation value of each term
in the lattice over all realisations.  Since the expectation is a linear operator, this yields a set
of entropies that are also idempotent, commutative, associative, absorptive and distributive, only
now over the random variables from~$\bm{X}$.  Thus, the information that Eve expects to gain from a
single realisation for a particular method of information sharing also corresponds to a term in a
free distributive lattice generated by $n$.  This distributive lattice for entropies can be seen in
\fig{distributive_lattices} by replacing $x$, $y$, $z$ and $h$ with $X$, $Y$, $Z$ and $H$
respectively.

Crucially, however, Property~\ref{prop:connex} does not hold for the entropies---it is not true that
Eve's expected information $H(X \sqcup Y)$ is given by either Alice's expected information $H(X)$ or
Bob's expected information $H(Y)$.  Thus, despite the fact that the distributive lattice of
information content can be reduced to a total order, the distributive lattice of entropies remains
partially ordered.  Although the information contents are totally ordered for every realisation,
this order is not in general the same for every realisation.  Consequently, when taking the
expectation value across many realisations to yield the corresponding entropies, the total order is
not maintained, and hence we are left with a partially ordered set of entropies.  Indeed, we already
saw the consequences of this result in \fig{union_info} whereby Alice's and Bob's information
content was totally ordered for any one realisation, but their expected information was partially
ordered.

\section{Multivariate Information Decomposition}
\label{sec:mid}

In \secRef{synergy}, we used the shared marginal information from \secRef{marginal} to decompose the
joint information content into four distinct components.  Our aim now is use the generalised notion
of shared information from the previous section to produce a generalised decomposition of the joint
information content.  To begin, suppose that Johnny observes the joint realisation $(x,y,z)$ while
Alice, Bob and Charlie observe the marginal realisations $x$, $y$ and $z$ respectively, and say that
Alice, Bob and Charlie share their information with Eve such that her information is given by
$h(x \sqcup y \sqcup z)$.  Clearly, Johnny has at least as much information as Eve,
\begin{equation}
  h(x, y, z) \geq h(x \sqcup y \sqcup z).
\end{equation}
Thus, we can compare how much more information Johnny has relative to Eve,
\begin{equation}
  \label{eq:synergistic_info_content_three}
  h(x \oplus y \oplus z) = h(x,y,z) - h(x \sqcup y \sqcup z)
  = \min\big( h(y,z|x), h(x,z|y), h(x,y|z) \big) \geq 0.
\end{equation}
This non-negative function generalises the earlier definition of the synergistic information content
\eq{synergistic_ic} such that it now quantifies how much information one gets from knowing the joint
probability $p_{XYZ}(x,y,z)$ relative to merely knowing the three marginal probabilities $p_X(x)$,
$p_Y(y)$ and $p_Z(z)$.  \fig{union_three} shows how this relationship can be represented using a
Venn diagram.

\begin{figure}[t]
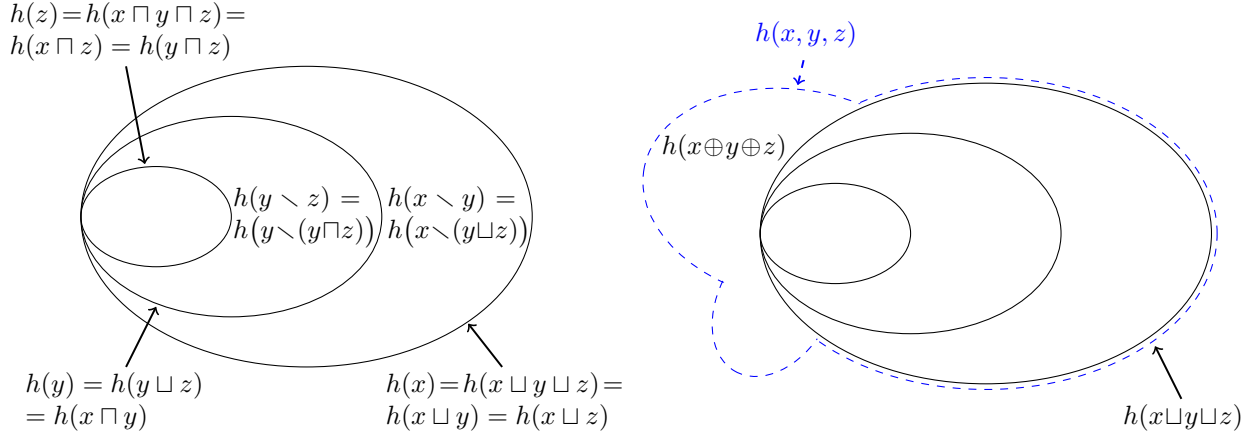

  \includegraphics{figs/union_entropy_three.tikz}
  \includegraphics{figs/synergistic_entropy_three.tikz}
  \caption{Similar to \fig{synergistic_info_content}, this Venn diagram shows how the synergistic
    information $h(x \oplus y \oplus z)$ can be defined by comparing the joint information content
    $h(x,y,z)$ to the union information content $h(x \sqcup y \sqcup z)$.  Note that, for this
    particular realisation, we are assuming that $h(x) > h(y) > h(z)$.}
  \label{fig:union_three}
\end{figure}

Now consider three more observers, Joan, Jonas, and Joanna, who observer the joint marginal
realisations $(x,y)$, $(x,z)$ and $(y,z)$, respectively.  Clearly, these additional observers
greatly increase the number of distinct ways in which marginal information might be shared with Eve.
For example, if Alice and Joanna share their information, then Eve's information is given by
$h\big(x \sqcup (y,z)\big)$.  Alternatively, if Joan and Jonas share their information, then Eve's
information is given by $h\big((x,y) \sqcup (x,z)\big)$.  Perhaps most interestingly, if Joan, Jonas
and Joanna share their information, then Eve's information is given by
$h\big((x,y) \sqcup (x,z) \sqcup (y,z)\big)$.  Moreover, we know that Johnny has at least as much
information as Eve has in this situation,
\begin{equation}
  h(x,y,z) \geq h\big((x,y) \sqcup (x,z) \sqcup (y,z)\big).
\end{equation}
Thus, by comparing how much more information Johnny has relative to Eve in this situation, we can
define a new type of synergistic information content that quantifies how much information one gets
from knowing the full joint realisation to merely knowing all of the pairwise marginal realisations,
\begin{equation}
  h\big((x,y) \oplus (x,z) \oplus (y,z)\big)
    = h\big(x,y,z\big) - h\big((x,y) \sqcup (x,z) \sqcup (y,z)\big)
    = \min\big(h(z|x,y),h(y|x,z),h(x|y,z)\big).
\end{equation}

\begin{samepage}
  Of course, these new ways to share joint information are not just restricted to the union
  information.  If Alice and Joanna share their information, then the information that Eve could
  have gotten from either is given by $h\big(x \sqcap (y,z)\big)$.  It is also worthwhile noting
  that this quantity is not less than the information that Eve could have gotten from either Alice's
  information or Bob’s and Charlie’s information,
  \begin{equation}
    h\big(x \sqcap (y,z)\big) \geq  h\big(x \sqcap (y \sqcup z)\big).
  \end{equation}
  Thus, we can also consider defining new types of synergistic information content associated with
  these this mixed type comparisons,
  \begin{equation}
    h\big(x \sqcap (y \oplus z)\big) = h\big(x\sqcap(y,z)\big) - h\big(x\sqcap(y \sqcup z)\big). 
\end{equation}
However, it is important to note that this quantity does not equal
$\min\big(h(x),\min(h(z|y),h(y|z))\big)$.
\end{samepage}

With all of these new ways to share joint marginal information, it is not immediately clear how we
should decompose Johnny's information.  Nevertheless, let us begin by considering the algebraic
structure of joint information content.  From the inequality \eq{shannon_inequalities_pw}, we know
that any pair of marginal information contents $h(x)$ and $h(y)$ are upper-bounded by the joint
information content $h(x,y)$.  It is also easy to see that the joint information content is
idempotent, commutative and associative.  Together, these properties are sufficient for establishing
that the algebraic structure of joint information content is that of a join
semi-lattice~\citep{gratzer2002} which we will denote by $\langle \bm{x}; h(\,,) \rangle$.
\fig{semi-lattices} shows the semi-lattices generated by $n=2$ and $n=3$ observers.

We now wish to establish the relationship between this semi-lattice of joint information content
$\langle \bm{x}; h(\,,) \rangle$ and the distributive lattice of shared marginal information
$\langle \bm{x}; h(\,\sqcup\,), h(\,\sqcap\,) \rangle$.  In particular, since our aim is to decompose
Johnny's information, consider the relationship between the join semi-lattice
$\langle \bm{x}; h(\,,) \rangle$ and the meet semi-lattice $\langle \bm{x}; h(\,\sqcap\,) \rangle$
which is also depicted in \fig{semi-lattices}.  In contrast to the semi-lattice of union information
content $\langle \bm{x}; h(\,\sqcup\,) \rangle$, the semi-lattice $\langle \bm{x}; h(\,,) \rangle$ is
not connected to the semi-lattice $\langle \bm{x}; h(\,\sqcap\,) \rangle$.  Although the intersection
information content absorbs the joint information content, since
\begin{equation}
 h\big(x \sqcap(x, y)\big) = h(x) \label{eq:capWithJoint}
\end{equation}
for all $h(x)$ and $h(y)$, the joint information content does not absorb the intersection
information content since $h\big(x, (x \sqcap y)\big)$ is equal to $h(x,y)$ for $h(x) \geq h(y)$,
i.e.\ is not equal to $h(x)$ as required for absorption.  Since the the join semi-lattice
$\langle \bm{x}; h(\,,) \rangle$ is not connected to the meet semi-lattice
$\langle \bm{x}; h(\,\sqcap\,) \rangle$ by absorption, their combined algebraic structure is not a
lattice.

\begin{figure}[t]
  \hspace{-10pt}
  \includegraphics{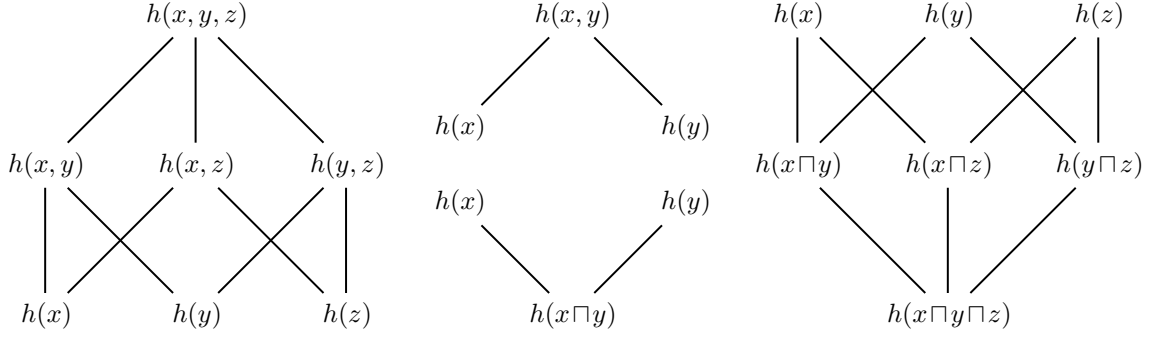}
  \caption{\emph{Top-middle and left}: The join semi-lattice
    $\big\langle \bm{h}\,; (\,,) \big\rangle$ for $n=2$ and $n=3$ marginal observers.  Johnny's
    information is always given by the joint information content at the top of the semi-lattice,
    while the information content of individuals such as Alice, Bob and Charlie who observe single
    realisations are found at the bottom of the semi-lattice.  The information content of joint
    marginal observers such as Joanna, Jonas and Joan are found in between these two extremities.
    \emph{Bottom-middle and right}: The meet semi-lattice $\big\langle \bm{h}\,; \sqcap \big\rangle$
    for $n=2$ and $n=3$ marginal observers.  Since these two semi-lattices are not connect by
    absorption, their combined structure is not a lattice.}
  \label{fig:semi-lattices}
\end{figure}

Despite the fact that the overall algebraic structure is not a lattice, there is a lattice
sub-structure $\langle \mc{A}(\bm{x}), \preceq \rangle$ within the general structure.  This
substructure is isomorphic to the redundancy lattice from the partial information
decomposition~\citep{williams2010} (see also \citep{lizier2018}), and its existence is a consequence
of the fact that the intersection information content absorbs the joint information content in
\eq{capWithJoint}.  In order to identify this lattice, we must first determine the reduced set of
elements $\mc{A}(\bm{x})$ upon which it is defined.  We begin by considering the set of all possible
joint realisations which is given by $\mc{P}_1(\bm{x})$ where
$\mc{P}_1(\bm{x}) = \mc{P}(\bm{x}) \sm \emptyset$.  Elements of this set $\mc{P}_1(\bm{x})$
correspond to the elements from the join semi-lattice $\langle \bm{x}; h(\,,) \rangle$, e.g.\ the
elements $\{x\}$ and $\{x,y\}$ correspond to $h(x)$ and $h(x,y)$, respectively.  In alignment with
\citet{williams2010}, we will call the elements of $\mc{P}_1(\bm{x})$ \emph{sources} and denote them
by $\bm{A}_1,\bm{A}_2,\ldots,\bm{A}_k$.  Next, we consider set of all possible \emph{collections of
  sources} which are given by the set $\mc{P}_1(\mc{P}_1(\bm{x}))$.  Each collection of sources
corresponds to an element of the meet semi-lattice
$\langle \mc{P}_1(\bm{x}) ; h(\,\sqcap\,) \rangle$, or a particular way in which we can evaluate the
intersection information content of a group of joint information contents.  For example, the
collections of sources $\{\{x\},\{y\}\}$ and $\{\{x\},\{y,z\}\}$ correspond to the $h(x \sqcap y)$
and $h\big(x \sqcap (y,z)\big))$, respectively.  Not all of these collections of sources are
distinct, however.  Since the intersection information content absorbs the joint information
content, we can remove the element $\{\{x\},\{x,y\}\}$ corresponding to $h\big(x \sqcap (x,y)\big)$
as this information is already captured by the element $\{\{x\}\}$ corresponding to $h(x)$.  In
general, we can remove any collection of sources that corresponds to the intersection information
content between a source $\bm{A}_i$ and any source $\bm{A}_j$ that is in the down-set
$\downarrow\!  \bm{A}_i$ with respect to the join semi-lattice
$\big\langle \bm{x}; h(\,,) \big\rangle$.  (A definition of the down-set can be found in
\citep{davey2002}. Informally, the down-set $\downarrow\! \bm{A}$ is the set of all elements that
precede $\bm{A}$.) By removing all such collections of sources, we get the following reduced set of
collections of sources,
\begin{equation}
  \label{eq:red_lattice}
  \mc{A}(\bm{x}) = \{ \alpha \in \mc{P}_1(\mc{P}_1(\bm{x})) :
    \forall \, \bm{A}_i,\, \bm{A}_j \in \alpha, \, \bm{A}_i \not\subset \bm{A}_j \}.
\end{equation}
Formally, this set corresponds to the set of antichains on the lattice
$\langle\mc{P}_1(\bm{x}),\subseteq \rangle$, excluding the empty set \citep{williams2010}.

Now that we have determined the elements upon which the lattice sub-structure is defined, we must
show that they indeed form a lattice.  Recall that when constructing the set $\mc{A}(\bm{x})$, we
first considered the ordered elements of the semi-lattice
$\big\langle \bm{x}; h(\,,) \big\rangle$ and then subsequently consider the ordered elements of the
semi-lattice $\big\langle \mc{P}_1(\bm{x}) ; h(\,\sqcap\,) \big\rangle$.  Thus, we need to show
that these two orders can be combined together into one new ordering relation over the set
$\mc{A}(\bm{x})$.  This can be done by extending the approach underlying the construction of the set
$\mc{A}(\bm{x})$ to consider any pair of collections of sets $\alpha$ and $\beta$ from
$\mc{A}(\bm{x})$.  In particular, the collection of sets $\beta$ precedes the collection of sets
$\alpha$ if and only if for every source $\bm{B}$ from $\beta$, there exists a source $\bm{A}$ from
$\alpha$ such that $\bm{A}$ is in the down-set $\downarrow\! \bm{B}$ with respect to the
join-semi-lattice $\big\langle \bm{x}; h(\,,) \big\rangle$, or formally,
\begin{equation}
  \label{eq:red_lattice_order}
  \forall \, \alpha, \beta \in \mc{A}(\bm{x}),
  (\alpha \preceq \beta \iff \forall \, \bm{B} \in \beta,
    \exists \,\bm{A} \in \alpha, \bm{A} \subseteq \bm{B}).
\end{equation}
The fact that $\langle\mc{A}(\bm{x}),\preceq\rangle$ forms a lattice was proved by
\citet{crampton2001,crampton2000} where the corresponding lattice is denoted
$\langle\mc{A}(X), \preceq^\prime\rangle$ in their notation.  Furthermore, they showed that this
lattice is isomorphic to the distributive lattices, and hence the number of elements in the set
$\mc{A}(\bm{x})$ for $n$ marginal observers is also given by the ($n$)-th Dedekind number
\citep[p.273]{comtet2012} (see also \citep{oeis}).  \citeauthor{crampton2000} also provided the meet
$\wedge$ and join $\vee$ operations for this lattice, which are given by
\begin{align}
  \alpha \wedge \beta &= \underline{\alpha \sqcup \beta},  \\
  \alpha \vee \beta &= \underline{\uparrow \! \alpha \cap \uparrow \! \beta},
\end{align}
where $\underline{\alpha}$ denotes the set of minimal elements of $\alpha$ with respect to the
semi-lattice $\big\langle \bm{x}; h(\,,) \big\rangle$.  (A definition of the set of minimal elements
can be found in \citep{davey2002}. Informally, $\underline{\alpha}$ is the set of sources of
$\alpha$ that are not preceded by any other sources from $\alpha$ with respect to the semi-lattice
$\big\langle \bm{x}; h(\,,) \big\rangle$.)  This lattice $\langle\mc{A}(\bm{x}),\preceq\rangle$ is
the aforementioned sub-structure that is isomorphic to the \emph{redundancy lattice} from
\citet{williams2010}.  However, as it is a lattice over information contents, it is actually
equivalent to the specificity lattice from \citep{finn2018pointwise}.  \fig{red_lattice} depicts the
redundancy lattice of information contents for $n=2$ and $n=3$ marginal observers.

\begin{figure}[p]
  \hspace{-35pt}
  \includegraphics{figs/red_lattices.tikz}
  \caption{ \emph{Top left}: The redundancy lattices $\langle \mc{A}(\bm{x}), \preceq \rangle$ of
    information contents for two and three and three observers.  Each note in the lattice
    corresponds to an element in $\mc{A}(\bm{x})$ from \eq{red_lattice}, while the ordering between
    elements is given by $\preceq$ from \eq{red_lattice_order}.  \emph{Bottom right}: The partial
    information contents $h_\partial(\alpha)$ corresponding to the redundancy lattices of
    information contents for two and three observers.}
  \label{fig:red_lattice}
\end{figure}

\begin{samepage}
  Similar to how Eve's information is non-decreasing as we move up through the terms of the
  distributive lattice of shared information, the redundancy lattice of information contents enables
  to see that, for example, the information that Eve could have gotten from either Alice or Joanna
  $h\big(x \sqcap (y,z)\big)$ is no less than the information that Eve could have gotten from Alice
  or Bob $h(x \sqcap y)$.  Thus, by taking the information $h(\alpha)$ associated with the
  collection of sources $\alpha$ from $\mc{A}(\bm{x})$ and subtracting from it the information
  $h(\alpha_i)$ associated with any collection of sources $\alpha_i$ from the down-set
  $\downarrow \! \alpha$, we can evaluate the unique information $h(\alpha \ssm \alpha_i)$ provided
  by $\alpha$ relative to $\alpha_i$.  Moreover, as per \citet{williams2010}, we can derive a
  function that quantifies the \emph{partial information content} $h_\partial(\alpha)$ associated
  with the collection of sources $\alpha$ that is not available in any of the collections of sources
  that are covered by $\alpha$.  (The set of collections of sources that are covered by $\alpha$ is
  denoted $\alpha^-$. A definition of the covering relation is provided in \citep{davey2002}.
  Informally, $\alpha^-$ is the set collections of sources that immediately precede $\alpha$.)
  Formally, this function corresponds to the M\"obius inverse of $h$ on the redundancy lattice
  $\langle \mc{A}(\bm{x}), \preceq \rangle$, and can be defined implicitly by
\begin{equation}
  \label{eq:implicit}
  h(\alpha) = \sum_{\beta \preceq \alpha} h_\partial(\beta).
\end{equation}
By subtracting away the partial information terms that strictly precede $\alpha$ from both sides, it
is easy to see that the partial information content $h_\partial(\alpha)$ can be calculated
recursively from the bottom of the redundancy lattice of information contents,
\begin{equation}
  \label{eq:implicit2}
  h_\partial(\alpha) = h(\alpha) - \sum_{\beta \prec \alpha} h_\partial(\beta),
\end{equation}
\end{samepage}

As the following theorem shows, the partial information content $h_\partial(\alpha)$ can be written
in closed-form.

\begin{theorem}
  \label{thm:closed_form}
  The partial information content $h_\partial(\alpha)$ is given by
  \begin{align}
    \label{eq:closed_form}
    h_\partial(\alpha)
      &= h(\alpha) - h(\alpha^-_1 \sqcup \alpha^-_2 \sqcup \ldots \sqcup \alpha^-_{|\alpha^-|}) \nn\\
      &= h(\alpha) - \max\big(h(\alpha^-_1),h(\alpha^-_2),\ldots,h(\alpha^-_{|\alpha^-|})\big), 
  \end{align}
  where each $\alpha^-_i$ is a collection of sets from $\alpha^-$.
\end{theorem}
\begin{proof}
  For $\bm{S} \subseteq \mc{A}(\bm{x})$, define the set-additive function
  \begin{equation}
    \label{eq:saf}
    f(\bm{S}) = \sum_{\beta \in \bm{S}} h_\partial(\beta).
  \end{equation}
  From \eq{implicit}, we have that $h(\alpha) = f(\downarrow \! \alpha)$.  The partial information
  can then by subtracting the set additive on the down-set $\downarrow \! \alpha$ from the set
  additive function on the strict down-set $\dot{\downarrow} \alpha$,
  \begin{equation}
    h_\partial(\alpha) = f(\downarrow \! \alpha) - f(\dot{\downarrow}\, \alpha)
      = f(\downarrow \! \alpha) - f\big( \bigcup_{\beta \in \alpha^-} \downarrow \! \beta \big),  
  \end{equation}
  By applying the principle of inclusion-exclusion \citep{stanley1997}, we get that
  \begin{equation}
    h_\partial(\alpha) = f(\downarrow \! \alpha)
      - \sum_{k=1}^{|\alpha^-|} (-1)^{k-1} \sum_{\substack{\bm{S}\subseteq\bm{\alpha^-} \\
    |\bm{S}|=k }} f\Big( \bigcap_{\sigma \in \bm{S}} \downarrow \! \sigma \Big). 
  \end{equation}
  For any lattice $L$ and $A \subseteq L$, we have that $\bigcap_{a \in A} \downarrow \! a$ is equal
  to $\downarrow \! ( \bigwedge A )$ \citep[p.57]{davey2002}, and since the meet operation is given
  by the intersection information content, we have that
  \begin{align}
    h_\partial(\alpha) &= f(\downarrow \! \alpha)
      - \sum_{k=1}^{|\alpha^-|} (-1)^{k-1} \sum_{\substack{\bm{S}\subseteq\bm{\alpha^-} \\
      |\bm{S}|=k }} h(s_1 \sqcap s_2 \sqcap \ldots \sqcap s_k) \nn\\
    &= h(\alpha) - h(\alpha_1 \sqcup \alpha_2 \sqcup \ldots \sqcup \alpha_{|\alpha^-|}), 
  \end{align}
  where the final step has been made using \eq{pie_union} and \eq{saf}.
\end{proof}

The closed-form solution \eq{closed_form} from Theorem~\ref{thm:closed_form} is the same as the
closed-form solution presented in Theorem~A2 from \citet{finn2018pointwise}.  This, together with
the aforementioned fact that the lattice $\langle\mc{A}(\bm{x}),\preceq\rangle$ is equivalent to the
specificity lattice, means that each partial information content $h_\partial(\alpha)$ is equal to
the partial specificity $i^+_\partial(\alpha \rightarrow t)$ from (A22) of
\citep{finn2018pointwise}.  As such, the partial information decomposition present in this paper
is equivalent to the pointwise partial information decomposition presented in
\citep{finn2018pointwise}.

Let us now use the closed-form solution \eq{closed_form} from Theorem~\ref{thm:closed_form} to
evaluate the partial information contents for the $n=2$ redundancy lattice of information contents.
Starting from the bottom, we get the intersection information content,
\begin{equation}
  h_\partial(x \sqcap y) = h(x \sqcap y), 
\end{equation}
followed by the unique information contents,
\begin{alignat}{4}
  &h_\partial(x) &&= h(x) &&- h(x \sqcap y) &&= h(x \ssm y), \\
  &h_\partial(y) &&= h(y) &&- h(x \sqcap y) &&= h(y \ssm x),
\end{alignat}
and finally, the synergistic information content, 
\begin{equation}
  h_\partial(x,y) = h(x,y) - h(x \sqcup y) = h(x \oplus y). 
\end{equation}
It is clear that these partial information contents recover the intersection, unique and synergistic
information contents from Sections \ref{sec:marginal} and \ref{sec:synergy}.  Moreover, by inserting
these partial terms back into \eq{implicit} for $\alpha=\{\{x,y\}\}$, we recover the earlier
decomposition \eq{bivar_decomp} of Johnny's information,
\begin{align}
  h(x,y) = h_\partial(x \sqcap y) + h_\partial(x) + h_\partial(y) + h_\partial(x,y)
  = h(x \sqcap y) + h(x \ssm y) + h(y \ssm x) + h(x \oplus y).
\end{align}

Of course, our aim is to generalise this result such that we can decompose the joint information
content for an arbitrary number of marginal realisations.  This can be done by first evaluating the
partial information contents over the redundancy lattice corresponding to $n$ marginal realisations,
and then subsequently inserting the results back into \eq{implicit} for
$\alpha=\{\{x_1, x_2, \ldots ,x_n\}\}$.  For example, we can invert the $n=3$ redundancy lattice of
information contents which yields the partial information contents shown in \fig{red_lattice}.  (The
inversion is evaluated in the Appendix.)  When inserted back into \eq{implicit}, we get the
following decomposition for Johnny's information,
\begin{align}
  \label{eq:trivar_decomp}
  h(x,y,z)
    &= h(x \sqcap y \sqcap z) \nn\\
    &\quad + h\big((x \sqcap y) \ssm z\big) + h\big((x \sqcap z) \ssm y\big)
      + h\big((y \sqcap z) \ssm x\big) \nn\\
    &\quad + h\big(x \sqcap (y \oplus z)\big) + h\big(y \sqcap (x \oplus z)\big)
      + h\big(z \sqcap (x \oplus y)\big) \nn\\
    &\quad + h\big(x \ssm (y,z)\big) + h\big(y \ssm (x,z)\big) + h\big(z \ssm (x,y)\big)
    + h\big((x \oplus y) \sqcap (x \oplus z) \sqcap (y \oplus z)\big) \nn\\
    &\quad + h\big((x \oplus y) \sqcap (x \oplus z) \ssm (y,z)\big)
      + h\big((x \oplus y) \sqcap (y \oplus z) \ssm (x,z)\big)
      + h\big((x \oplus z) \sqcap (y \oplus z) \ssm (x,y)\big) \nn\\
    &\quad + h\big((x \oplus y) \ssm ((x,z) \sqcup (y,z))\big)
      + h\big((x \oplus z) \ssm ((x,y) \sqcup (y,z))\big)
      + h\big((y \oplus z) \ssm ((x,y) \sqcup (x,z))\big) \nn\\
  &\quad + h\big((x,y) \oplus (x,z) \oplus (y,z)\big).
\end{align}

Finally, we can also consider taking the expectation value of each term in the redundancy lattice of
information contents.  Since the expectation is a linear and monotonic operator, the resulting
expectation values will inherit the structure of the redundancy lattice of information contents and
so form a redundancy lattice of entropies, i.e.\ \fig{red_lattice} with $x$, $y$, $z$ and $h$
replaced by $X$, $Y$, $Z$ and $H$ respectively.  By inverting the $n=2$ redundancy lattice of
entropies, we can recover the decomposition \eq{bivar_decomp_avg} from \fig{synergistic_entropy}.
Furthermore, inverting the $n=3$ lattice generalises this result and is depicted in \fig{pid_3}.

\begin{figure}[p]
  \includegraphics{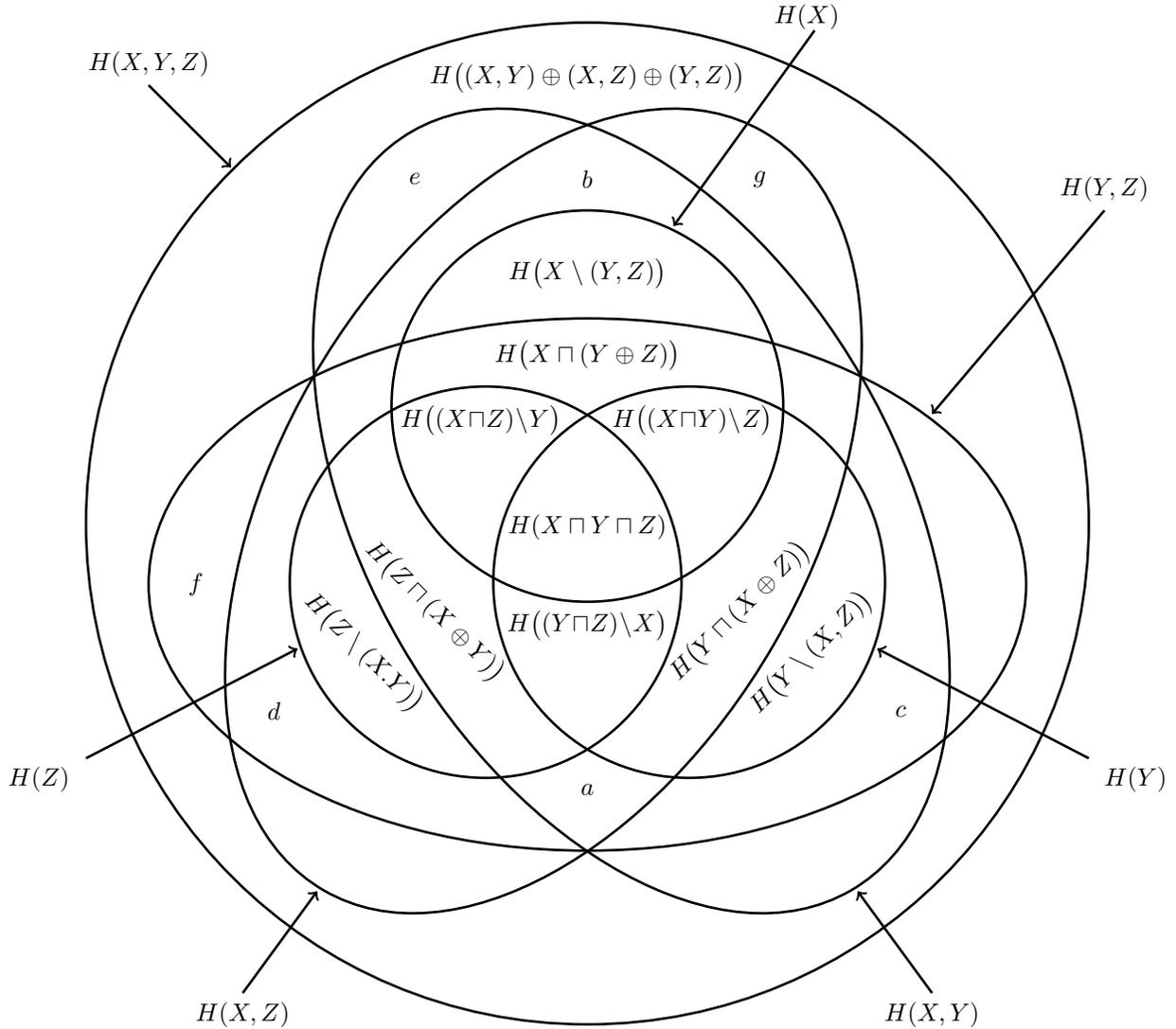}
  \caption{This Venn provides a visual representation of the decomposition of the joint entropy
    $H(X,Y,Z)$.  This decomposition is given by replacing $x$, $y$, $z$ and $h$ with $X$, $Y$, $Z$
    and $H$ in \eq{trivar_decomp}, respectively.}
  \label{fig:pid_3}
\end{figure}

\section{Union and Intersection Mutual Information}
\label{sec:info}

Suppose that Alice, Bob and Johnny are now additionally and commonly observing the variable $Z$.
When a realisation $(x,y,z)$ occurs, Alice's information given $z$ is given by the conditional
information content $h(x|z)$, while Bob's conditional information is given by $h(y|z)$ and Johnny's
conditional information is given by $h(x,y|z)$.  By using the same argument as in
\secRef{marginal}, it is easy to see that Eve's conditional information given $z$ is given by the
conditional union information content,
\begin{equation}
  \label{eq:cond_union_info_content}
  h(x \sqcup y | z) = \max\big(h(x|z), h(y|z) \big).  
\end{equation}
Likewise, we can define the conditional unique information contents and conditional intersection
information content, respectively,
\begin{alignat}{2}
  &h(x \ssm y | z) &&= h(x \sqcup y | z) - h(y | z) = \max\big(h(x|z) - h(y|z),0\big), \\
  &h(y \ssm x | z) &&= h(x \sqcup y | z) - h(x | z) = \max\big(0,h(y|z) - h(x|z)\big), \\
  &h(x \sqcap y|z)   &&= h(x|z) + h(y|z) - h(x \sqcup y|z) = \min\big(h(x|z), h(y|z)\big).
\end{alignat}
Furthermore, since Johnny's conditional information $h(x,y|z)$ is no less than Eve's conditional
information content $h(x \sqcup y | z)$, we can also define the conditional synergistic information
content,
\begin{equation}
  h(x \oplus y|z) = h(x,y|z) - h(x \sqcup y | z) = \min\big(h(y|x,z), h(x|y,z)\big).
\end{equation}
Similar to \eq{bivar_decomp}, we can decompose Johnny's conditional information $h(x,y|z)$ into the
following components,
\begin{equation}
  \label{eq:bivar_decomp_cond}
  h(x,y|z) = h(x \sqcup y|z) + h(x \oplus y|z)
  = h(x \sqcap y|z) + h(x \ssm y|z) + h(y \ssm x|z) + h(x \oplus y|z).
\end{equation}
Moreover, similar to \eq{mi_diff}, the conditional mutual information content is equal to the
difference between the conditional intersection information content and the conditional synergistic
information content,
\begin{equation}
  \label{eq:mi_diff_cond}
  i(x;y|z) = h(x|z) + h(y|z) - h(x,y|z) = h(x \sqcap y|z) - h(x \oplus y|z).
\end{equation}
Notice that all of the above definitions directly correspond to the definitions of the unconditioned quantities, with all probability distributions conditioned on $z$ here.


Let us now consider how much information each of our observers have about the commonly observed
realisation~$z$.  The information that Alice has about $z$ from observing $x$ is given by the mutual
information content,
\begin{equation}
  \label{eq:alice_mi}
  i(x;z) = h(x) - h(x|z).
\end{equation}
Similarly, Bob's information about $z$ is given by $i(y;z)$, while Johnny's information is given by
the joint mutual information content $i(x,y;z)$.  Thus, the question naturally arises---are we able
to quantify how much information Eve has about the realisation $z$ from knowing Alice's and Bob's
shared information?

Clearly, we could consider defining the union mutual information content,
\begin{equation}
  \label{eq:union_mic}
  i(x \sqcup y ; z) = h(x \sqcup y) - h(x \sqcup y | z).
\end{equation}
It is important to note that, while the mutual information can be defined in three different ways
$i(x,z) = h(x) - h(x|z) = h(x) + h(z) - h(x,z) = h(z) - h(z|x)$, there is only one way in which one
can define this function.  (Indeed, this point aligns well with our argument based on exclusions
presented in \citep{finn2018probability}.)  Similar to \eq{union_mic}, we could consider
respectively defining the unique mutual information contents, the intersection mutual information
content and synergistic mutual information content,
\begin{align}
  i(x \ssm y ; z) &= h(x \ssm y) - h(x \ssm y | z), \\
  i(y \ssm x ; z) &= h(y \ssm x) - h(y \ssm x | z), \\
  i(x \sqcap y ; z) &= h(x \sqcap y) - h(x \sqcap y | z), \\ 
  i(x \oplus y ; z) &= h(x \oplus y) - h(x \oplus y | z). \label{eq:synergistic_mic}
\end{align}
Just like the mutual information content \eq{mutual_info_content}, there is nothing to suggest that
these quantities are non-negative.  Of course, the mutual information or expected mutual information
content \eq{mutual_info} is non-negative.  Thus, with this in mind, consider defining the union
mutual information
\begin{equation}
  \label{eq:union_mi}
  I(X \sqcup Y ; Z) = \mathrm{E}_{XYZ} \big[i(x \sqcup y;z)\big].
\end{equation}
However, there is nothing to suggest that this function is non-negative.  Consequently, it is
dubious to claim that this function represents Eve's expected information about $Z$, and is
similarly fallacious to say that Eve's information about $z$ is given by the union mutual
information content \eq{union_mic}.  Indeed, by inserting the definitions \eq{union_info_content}
and \eq{cond_union_info_content} into \eq{union_mic}, it is easy to see why it is difficult to
interpret these functions,
\begin{align}
  i(x \sqcup y ; z) &= \max\big(h(x),\, h(y)\big) - \max\big(h(x|z),\, h(y|z)\big) \nn\\
    &= \max\Big(\min\big(i(x;z),\,h(x)-h(y|z)\big),\,\min\big(h(y)-h(x|z),\,i(y;z)\big)\Big).
\end{align}
That is, the union mutual information content can mix the information content provided by one
realisation with the conditional information content provided by another.  Thus, there is no
guarantee that this function's expected value will be non-negative.  It is perhaps best to interpret
this function as being a difference between two surprisals, rather than a function which represent
information.  Of course, similar to the multivariate mutual information \eq{mmi}, the union
mutual information can be used a summary quantity provided one is careful not to misinterpret its
meaning.  The same is true for the unique mutual informations,
intersection mutual information and synergistic mutual information, which we can similarly define,
\begin{align}
  I(X \sm Y ; Z) &= \mathrm{E}_{XYZ} \big[i(x \ssm y;z)\big], \\
  I(Y \sm X ; Z) &= \mathrm{E}_{XYZ} \big[i(y \ssm x;z)\big], \\
  I(X \sqcap Y ; Z) &= \mathrm{E}_{XYZ} \big[i(x \sqcap y;z)\big], \\
  I(X \oplus Y ; Z) &= \mathrm{E}_{XYZ} \big[i(x \oplus y;z)\big].
\end{align}

Despite lacking the clear interpretation that we had for the information contents, these functions
share a similar algebraic structure.  For example, by using \eq{bivar_decomp} and
\eq{bivar_decomp_cond}, we can decompose the mutual information content into the following
components,
\begin{equation}
  i(x,y;z) =  i(x \sqcap y ; z) + i(x \ssm y ; z) + i(y \ssm x ; z) +  i(x \oplus y ; z), 
\end{equation}
which is similar to the earlier decomposition of the joint entropy \eq{bivar_decomp}.  Moreover,
similar to \eq{mi_diff}, by using \eq{mi_diff} and \eq{mi_diff_cond}, we get that the multivariate
mutual information content is given by the difference between the intersection mutual information
content and the synergistic mutual information content,
\begin{align}
  i(x;y;z) &= i(x;y) - i(x;y|z)
             = h(x \sqcap y) - h(x \oplus y) - h(x \sqcap y) + h(x \oplus y|z) \nn\\
           &= i(x \sqcap y; z) - i(x \oplus y; z).
\end{align}
Of course, since the expectation value is a linear operator, both of these results can be carried
over to the joint mutual information.  Hence, the mutual information can be decomposed into the
following components,
\begin{equation}
  I(X,Y;Z) =  I(X \sqcap Y ; Z) + I(X \sm Y ; Z) + I(Y \sm X ; Z) +  I(X \oplus Y ; Z).
\end{equation}
while the the multivariate mutual information is equal to the intersection mutual information minus
the synergistic mutual information,
\begin{align}
  I(X;Y;Z) &= I(X;Y) - I(X;Y|Z)
             = H(X \sqcap Y) - H(X \oplus Y) - H(X \sqcap Y) + H(X \oplus Y|Z) \nn\\
           &= I(X \sqcap Y; Z) - I(X \oplus Y; Z).
\end{align}
This latter result aligns Williams and Beer's prior result that the multivariate mutual information
conflates redundant and synergistic information \citep[eq. 14]{williams2010}.

\section{Conclusion}
\label{conclusion}

The main aim of this paper has been to understand and quantify the distinct ways that a set of
marginal observers can share their information with some non-observing third party.  To accomplish
objective, we examined the distinct ways in which two marginal observers, Alice and Bob, can share
their information with the non-observing individual, Eve, and introduced several novel measures of
information content:\ the union, intersection and unique information contents.  We then investigated
the algebraic structure of these new measures of shared marginal information and showed that the
structure of shared marginal information is that of a distributive lattice.  Furthermore, by using
the fundamental theorem of distributive lattices, we showed that these new measures are isomorphic
to the set union and intersections.  This isomorphism is similar to Yeung's correspondence between
multivariate mutual information and signed measure \citep{yeung1991, yeung2008}.  However, in
contrast to Yeung's correspondence, the measures of information content presented in this paper are
non-negative.  Moreover, these measures maintain a clear operational meaning regardless of the
number of realisations or variables involved.  (This is, of course, excepting the mutual information
contents presented in \secRef{info} which are not non-negative; whether or not one should use such
quantities is open to debate.)

The secondary objective of this paper has been to understand and demonstrate how we can use the
measures of shared information content to decompose multivariate information.  We began by comparing
the union information content to the joint information content and used this comparison to define a
measure of synergistic information content.  We showed how one can use this measure, together with
the measures of shared information content, to decompose the joint information content.  Next, we
compared the algebraic structure of joint information to the lattice structure of shared
information, and showed how one can find the redundancy lattice from the partial information
decomposition \citep{williams2010} embedded within this larger algebraic structure.  More
specifically, since this paper considers information contents, this redundancy lattice is actually
same as the specificity lattice from pointwise partial information decomposition
\citep{finn2018probability, finn2018pointwise}.

This result links the work presented in this paper to the existing body of theoretical literature on
information decomposition \citep{williams2010, bertschinger2013, harder2013, harder2013phd,
  bertschinger2014, griffith2014a, griffith2014b, rauh2014, barrett2015, griffith2015, olbrich2015,
  perrone2016, rosas2016, faes2017, ince2017pid, james2017, kay2017, makkeh2017, pica2017, quax2017,
  rauh2017a, rauh2017b, rauh2017c, james2018}, and its applications~\citep{williams2011,
  flecker2011, lizier2013, stramaglia2014, timme2014, wibral2015, biswas2016, frey2016, timme2016,
  ghazi2017, maity2017, sootla2017, tax2017, wibral2017a, wibral2017b, finn2018quantifying,
  rosas2018, wollstadt2018idtxl, biswas2019, james2019, kolchinsky2019, li2019, rosas2019}.  (For a
brief summary of this literature, see \citep{lizier2018}.)  In contrast to the pointwise partial
information decomposition \citep{finn2018probability, finn2018pointwise}, most of these approaches
aim to decompose the average mutual information rather than the information content.  One exception
to this statement is due to \citet{ince2017ped}, who proposes a method of information decomposition
based upon the entropy.  Moreover, Ince discusses how this decomposition can be applied to the
information content (or in Ince's terminology, the local entropy).  Of particular relevance to this
paper, Ince obtains a result that is equivalent to \eq{mi_diff} whereby the mutual information
content is equal to the redundant information content minus the synergistic information content
\citep[eq.5]{ince2017ped}.  However, Ince's definition of redundant information content differs from
that of the intersection information content in \eq{mi_diff}.  To be specific, it is based upon the
sign of the multivariate mutual information content (or pointwise co-information), which is
interpreted as a measure of ``the set-theoretic overlap'' of multiple information contents (or local
entropies)~\citep[p.7]{ince2017ped}.  However, as discussed in \secRef{intro}, this set-theoretic
interpretation of the multivariate mutual information (co-information) is problematic.  To account
for these difficulties, Ince disregards the negative values, defining the redundant information
content to equal to the multivariate mutual information when it is positive, and to be zero
otherwise.



There are a number of avenues of inquiry for which this research will yield new insights,
particularly in complex systems, neuroscience, and communications theory.  For instance, these
measures might be used to better understand and quantify distributed intrinsic computation
\citep{lizier2013, finn2018quantifying}.  It is well known that that dynamics of individual regions
in the brain depend synergistically on multiple other regions; synergistic information content might
provide a means to quantify such dependencies in neural data~\citep{lizier2011multivariate,
  vakorin2009confounding, novelli2019large, wibral2015, wibral2017a}.  Furthermore, these measures
might be helpful for quantifying the synergistic encodings used in network coding \citep{yeung2008}.
Finally, it is well-known that many biological traits are not dependent on any one gene, but rather
are synergistically dependent on two or more genes, and the decomposed information provides a means
to quantify the unique, redundant and synergistic dependencies between a trait and a set of
genes~\citep{deutscher2008, anastassiou2007, white2011}.

\bibliography{intersection_entropy} 

\subsection*{Acknowledgements}
The authors would like to thank Nathan Harding, Richard Spinney, Daniel Polani, Leonardo Novelli and
Oliver Cliff for helpful discussions relating to this manuscript.

\appendix
\setcounter{secnumdepth}{0}
\section{Appendix}

We can use the closed-form \eq{closed_form} to evaluate the partial information contents for the
$n=3$ redundancy lattice of information contents.  Starting from the bottom and working up through
the redundancy lattice, we have the following partial information contents:
\begin{align}
  h_\partial(x \sqcap y \sqcap z)
    &= h(x \sqcap y \sqcap z); \displaybreak[0]\\[10pt]
  h_\partial(x \sqcap y)
    &= h(x \sqcap y) - h(x \sqcap y \sqcap z) \nn\\
    &= h\big((x \sqcap y) \ssm z\big); \displaybreak[0]\\[10pt]
  h_\partial\big(x \sqcap (y,z)\big)
    &= h\big(x \sqcap (y,z)\big) - h\big((x \sqcap y) \sqcup (x \sqcap z)\big)\nn\\
    &= h\big(x \sqcap (y,z)\big) - h\big(x \sqcap (y \sqcup z)\big) \nn\\
    &= h\big(x \sqcap (y \oplus z)\big); \displaybreak[0]\\[10pt]
  h_\partial(x)
    &= h(x) - h\big(x \sqcap (y,z)\big)\nn\\
    &= h\big(x \ssm (y,z)\big); \displaybreak[0]\\[10pt]
  h_\partial\big((x,y)\sqcap(x,z)\sqcap(y,z)\big)
    &= h\big((x,y)\sqcap(x,z)\sqcap(y,z)\big)
      - h\big((x\sqcap(y,z))\sqcup(y\sqcap(x,z))\sqcup(z\sqcap(x,y))\big) \nn\\
    &= h\big((x,y)\sqcap(x,z)\sqcap(y,z)\big) - h\big(x \sqcap (y \oplus z)\big)
      - h\big(y \sqcap (x \oplus z)\big) - h\big(z \sqcap (x \oplus y)\big) \nn\\
    &\qquad - h\big((x \sqcap y) \ssm z \big) - h\big((x \sqcap z) \ssm y \big)
      - h\big((y \sqcap z) \ssm x \big) - h(x \sqcap y \sqcap z) \nn\\
    &= h\big((x,y)\sqcap(x,z)\sqcap(y,z)\big) - h\big(x \sqcap (y, z)\big)
      - h\big(y \sqcap (x \oplus z)\big) - h\big(z \sqcap (x \oplus y)\big) \nn\\
    &\qquad - h\big((y \sqcap z) \ssm x \big) \tag{By~\lem{L9}.}  \nn\\
    &= h\big((x \oplus y) \sqcap (x \oplus z) \sqcap (y \oplus z)\big); \displaybreak[0]\\[10pt]
  h_\partial\big((x,y) \sqcap (x,z) \big)
    &= h\big((x,y) \sqcap (x,z)\big)
      - h(x \sqcup ((x,y) \sqcap (x,z) \sqcap(y,z))\big) \nn\\
    &= h\big((x,y) \sqcap (x,z)\big) -h(x) - h((x,y) \sqcap (x,z) \sqcap(y,z)\big) \nn\\
      &\qquad + h(x \sqcap ((x,y) \sqcap (x,z) \sqcap(y,z))\big) \nn\\
    &= h\big(((x,y)\sqcap(x,z)) \ssm (y,z)\big) - h\big(x \ssm (y,z)\big)
      \tag{By Lemma~\ref{lem:L1}.}\nn\\
    &= h\big( ((x,y) \ssm (y,z)) \sqcap ((x,z) \ssm (y,z)) \big) - h\big(x \ssm (y,z)\big)
      \tag{By Lemma~\ref{lem:L2}.}\nn\\
    &=h\big((x \oplus y) \sqcap (x \oplus z) \ssm (y,z)\big); \displaybreak[0]\\[10pt] 
  h_\partial(x,y)
    &= h(x,y) - h\big(((x,y)\sqcap(x,z)) \sqcup ((x,y)\sqcap(y,z))\big) \nn\\
    &= h(x,y) - h\big((x,y) \sqcap ((x,z) \sqcup (y,z))\big) \nn\\
    &= h\big((x,y) \ssm ((x,z) \sqcup (y,z))\big) \tag{By Lemma~\ref{lem:L3}.} \nn\\ 
    &= h\big((x \oplus y) \ssm ((x,z) \sqcup (y,z))\big); \displaybreak[0]\\[10pt]
  h_\partial(x,y,z)
    &= h(x,y,z) - h\big((x,y) \sqcup (x,z) \sqcup (y,z)\big) \nn\\
    &= h\big((x,y) \oplus (x,z) \oplus (y,z)\big).
\end{align} 

\begin{lemma} We have the following identity,
  \label{lem:L1}
  \begin{equation}
    h\big((x \sqcap y) \ssm z\big) = h\big((x \ssm z) \sqcap (y \ssm z)\big).
  \end{equation}
\end{lemma}

\begin{proof}
  From \eq{uni_ic_x} and \eq{max_min_two}, we have that
  \begin{align}
    h\big((x \sqcap y) \ssm z\big) 
      &= \max\big(\min(h(x),h(y))-h(z), 0\big) = \max\big(\min(h(x)-h(z),\;h(y)-h(z)), 0\big) \nn\\
      &= \min\big(\max(h(x)-h(z),0),\;\max(h(y)-h(z),0)\big) = h\big((x \ssm z)\sqcap(y \ssm z)\big), 
  \end{align}
  where we have used the fact that $\max(a,\,\min(b,c))$ is equal to $\min(\max(a,b),\,\max(a,c))$.
\end{proof}

\begin{lemma}
\label{lem:L2}
  We have the following identity,
  \begin{equation}
    h\big(((x,y) \ssm (y,z)) \sqcap ((x,z) \ssm (y,z)) \big)
      = h\big(x \ssm (y,z)\big) + h\big((x \oplus y) \sqcap (x \oplus z) \ssm (y,z)\big).
  \end{equation}
\end{lemma}

\begin{proof}
  From \eq{bivar_decomp}, we have that
  \begin{align}
    h\big((x,y) \ssm (y,z)\big)
      &= h\big((x \sqcap y) \ssm (y,z)\big) + h\big((x \ssm y) \ssm (y,z)\big) + h\big((y \ssm x) \ssm (y,z)\big) + h\big((x \oplus y) \ssm (y,z)\big) \nn\\
      &= h\big(x \ssm (y,z)\big) + h\big((y \ssm x) \ssm (y,z)\big)
        + h\big((x \oplus y) \ssm (y,z)\big),
  \end{align}
  and since
  \begin{equation}
    h\big((y \ssm x) \ssm (y,z)\big) = \max\big(h(y \ssm x)-h(y,z),0\big) = 0,
  \end{equation}
  we get that
  \begin{equation}
    h\big((x,y) \ssm (y,z)\big) = h\big(x \ssm (y,z)\big) + h\big((x \oplus y) \ssm (y,z)\big). 
  \end{equation}
  Finally, by inserting this result into \eq{max_min_two}, we get that
  \begin{multline}
    h\big(((x,y) \ssm (y,z)) \sqcap ((x,z) \ssm (y,z)) \\
    \begin{aligned}
      &= \min\big(h\big(x \ssm (y,z)\big) + h\big((x \oplus y) \ssm (y,z)\big), \,
        h\big(x \ssm (y,z)\big) + h\big((x \oplus z) \ssm (y,z)\big)\big) \nn\\
      &= h\big(x \ssm (y,z)\big)
        + h\big(((x \oplus y) \ssm (y,z)) \sqcap ((x \oplus z) \ssm (y,z))\big). 
    \end{aligned}
  \end{multline}
\end{proof}

\begin{lemma}
  \label{lem:L3}
  We have the following identity, 
  \begin{equation}
    h\big(((x,y) \ssm ((x,z) \sqcup (y,z)) \big) = h\big(((x \oplus y) \ssm ((x,z) \sqcup (y,z)) \big).
  \end{equation}
\end{lemma}

\begin{proof}
  By \eq{bivar_decomp}, we have that
  \begin{align}
    h\big((x,y) \ssm ((x.z) \sqcup (y,z))\big)
      &= h\big((x \sqcap y) \ssm ((x,z) \sqcup (y,z))\big)
        + h\big((x \ssm y) \ssm ((x,z) \sqcup (y,z))\big) \nn\\
      &\quad + h\big((y \ssm x) \ssm ((x,z) \sqcup (y,z))\big)
        + h\big((x \oplus y) \ssm ((x,z) \sqcup (y,z))\big). 
  \end{align}
  Since $h(x \sqcap y)$, $h(x \ssm y)$ and $h(y \ssm x)$ are less than $h\big(((x,z) \sqcup (y,z)\big)$,
  from \eq{uni_ic_x} we have that $h\big((x \sqcap y) \ssm ((x,z) \sqcup (y,z))\big)$,
  $h\big((x \ssm y) \ssm ((x,z) \sqcup (y,z))\big)$ and
  $h\big((y \ssm x) \ssm ((x,z) \sqcup (y,z))\big)$ are all equal to $0$.
\end{proof}

\begin{lemma}
  \label{lem:L4}
  We have the following identity, 
  \begin{equation}
    h\big(x \sqcap (y \ssm x)\big) = 0.
  \end{equation}
\end{lemma}

\begin{proof}
  We have that
  \begin{align}
    h\big(x \sqcap (y \ssm x)\big) = h\big(x \sqcap (x,y)\big) - h\big(x \sqcap x\big)
    = h(x \sqcap x) - h(x) = 0, 
  \end{align}
  where we have used \eq{intersectionIdempotent} and \eq{capWithJoint}.
\end{proof}

\begin{lemma}
  \label{lem:L5}
  We have the following identity, 
  \begin{equation}
    h\big((y \ssm x) \sqcap (y, z)\big) = h\big(y \ssm x \big).
  \end{equation}
\end{lemma}

\begin{proof}
  We have that
  \begin{align}
    h\big((y \ssm x) \sqcap (y, z)\big)
      &= h\big((y \ssm x) \sqcap (y \ssm x, y \sqcap x, z)\big). 
  \end{align}
  Using \eq{capWithJoint}, this then reduces via to
  \begin{align}
    h\big((y \ssm x) \sqcap (y, z)\big) = h\big(y \ssm x \big). 
  \end{align}
\end{proof}

\begin{lemma}
  \label{lem:L6}
  We have the following identity, 
  \begin{equation}
    h\big(x \sqcap (x \oplus z)\big) = 0. 
  \end{equation}
\end{lemma}

\begin{proof}
  We have
  \begin{align}
    h\big(x \sqcap (x \oplus z)\big)
      &= h\big(x \sqcap (x, z)\big) - h\big(x \sqcap x\big) + h\big(x \sqcap (z \ssm x)\big) \nn, 
  \end{align}
  which then reduces via \eq{capWithJoint}, \eq{intersectionIdempotent} and \lem{L4} to:
  \begin{align}
    h\big(x \sqcap (x \oplus z)\big)
      &= h\big(x\big) - h\big(x\big) + h\big(x \sqcap (z \ssm x)\big) + 0\\
      &=0. \nn \qedhere
  \end{align}
\end{proof}

\begin{lemma}
  \label{lem:L7}
  We have the following identity, 
  \begin{equation}
    h\big((y \ssm x) \sqcap (x \oplus z)\big) = h\big(y \sqcap (x \oplus z) \big).
  \end{equation}
\end{lemma}

\begin{proof}
  We have that
 \begin{align}
    h\big(y \sqcap (x \oplus z) \big)
      &= h\big((y \ssm x) \sqcap (x \oplus z)\big)
        + h\big((y \sqcap x) \sqcap (x \oplus z)\big) \nn. 
  \end{align}
  But since $h\big(x \sqcap (x \oplus z)\big) = 0$ via \lem{L6} and therefore $h\big((y \sqcap x) \sqcap (x \oplus z)\big) = 0$, then we have 
  \begin{align}
    h\big(y \sqcap (x \oplus z) \big) = h\big((y \ssm x) \sqcap (x \oplus z)\big)  \nn
  \end{align}
  as required.
\end{proof}

\begin{lemma}
  \label{lem:L8}
  We have the following identity, 
  \begin{equation}
    h\big((x \oplus y) \sqcap (x, z) \sqcap (y,z) \big) = h\big(z \sqcap (x \oplus y) \big) + h\big((x \oplus y) \sqcap (x \oplus z) \sqcap (y \oplus z) \big).
  \end{equation}
\end{lemma}

\begin{proof}
  We have that
  \begin{align}
    h\big((x \oplus y) \sqcap (x, z) \sqcap (y,z) \big)
      &= h\big((x \oplus y) \sqcap x \sqcap (y,z) \big)
        + h\big((x \oplus y) \sqcap (z \ssm x) \sqcap (y,z) \big) \nn\\
      &\qquad + h\big((x \oplus y) \sqcap (x \oplus z) \sqcap (y,z) \big). 
  \end{align}
  Then, by using \lem{L6}, we get that
  \begin{align}
    h\big((x \oplus y) \sqcap (x, z) \sqcap (y,z) \big)
      &= h\big((x \oplus y) \sqcap (z \ssm x) \sqcap (y,z) \big)
        + h\big((x \oplus y) \sqcap (x \oplus z) \sqcap (y,z) \big) \nn \\
      &= h\big((x \oplus y) \sqcap (z \ssm x) \sqcap y \big)
        + h\big((x \oplus y) \sqcap (z \ssm x) \sqcap (z \ssm y) \big) \nn\\
        &\qquad + h\big((x \oplus y) \sqcap (z \ssm x) \sqcap (y \oplus z) \big)
        + h\big((x \oplus y) \sqcap (x \oplus z) \sqcap y \big) \nn\\
        &\qquad + h\big((x \oplus y) \sqcap (x \oplus z) \sqcap (z \ssm y) \big)
        + h\big((x \oplus y) \sqcap (x \oplus z) \sqcap (y \oplus z) \big) \nn\\
      &= h\big((x \oplus y) \sqcap (z \ssm x) \sqcap (z \ssm y)
        + h\big((x \oplus y) \sqcap (x \oplus z) \sqcap (y \oplus z) \big),
  \end{align}
  where we have used \lem{L6} four more times.  Finally, using \lem{L7}, we get that
  \begin{align}
    h\big((x \oplus y) \sqcap (x, z) \sqcap (y,z) \big)
      &= h\big((x \oplus y) \sqcap z \big)
        + h\big((x \oplus y) \sqcap (x \oplus z) \sqcap (y \oplus z) \big),
  \end{align}
  as required.
\end{proof}

\begin{lemma}
  \label{lem:L9}
  We have the following identity, 
  \begin{align}
    h\big((x, y) \sqcap (x, z) \sqcap (y,z) \big)
      &= h\big(x \sqcap (y,z) \big) + h\big(y \sqcap (x \oplus z) \big)
        + h\big(z \sqcap (x \oplus y) \big) + h\big((y \sqcap z) \ssm x \big) \nn\\
      &\qquad + h\big((x \oplus y) \sqcap (x \oplus z) \sqcap (y \oplus z) \big).
  \end{align}
\end{lemma}

\begin{proof}
  We have that
  \begin{align}
    h\big((x, y) \sqcap (x, z) \sqcap (y,z) \big)
      &= h\big(x \sqcap (x, z) \sqcap (y,z) \big)
        + h\big((y \ssm x) \sqcap (x, z) \sqcap (y,z) \big) 
        + h\big((x \oplus y) \sqcap (x, z) \sqcap (y,z) \big) \nn\\
      &= h\big(x \sqcap (y,z) \big)
        + h\big((y \ssm x) \sqcap (x, z) \big) 
        + h\big((x \oplus y) \sqcap (x, z) \sqcap (y,z) \big),
  \end{align}
  where we have used \eq{capWithJoint} and \lem{L5}.  Next, we have that
  \begin{align}
    h\big((x, y) \sqcap (x, z) \sqcap (y,z) \big)
      &= h\big(x \sqcap (y,z) \big)
        + h\big((y \ssm x) \sqcap x \big) 
        + h\big((y \ssm x) \sqcap (z \ssm x) \big) 
        + h\big((y \ssm x) \sqcap (x \oplus z) \big) \nn\\
        &\qquad + h\big((x \oplus y) \sqcap (x, z) \sqcap (y,z) \big) \nn\\
      &= h\big(x \sqcap (y,z) \big)
        + 0 
        + h\big((y \sqcap z) \ssm x \big) 
        + h\big(y \sqcap (x \oplus z) \big) \nn\\
      &\qquad + h\big((x \oplus y) \sqcap (x, z) \sqcap (y,z) \big)
  \end{align}
  where we have used \lem{L1}, \lem{L4} and \lem{L7}.  Finally, we have that
  \begin{align}
    h\big((x, y) \sqcap (x, z) \sqcap (y,z) \big)
      &= h\big(x \sqcap (y,z) \big)
        + h\big((y \sqcap z) \ssm x \big) 
        + h\big(y \sqcap (x \oplus z) \big) 
        + h\big(z \sqcap (x \oplus y) \big) \nn\\
      &\qquad + h\big((x \oplus y) \sqcap (x \oplus z) \sqcap (y \oplus z) \big),
  \end{align}
  where we have used via \lem{L8} to get the required result.
\end{proof}

\end{document}